\documentclass[a4paper]{article}
\pdfoutput=1
\usepackage{fullpage}
\usepackage[utf8]{inputenc}
\usepackage[T1]{fontenc}
\usepackage[english]{babel}
\usepackage{graphicx} 
\usepackage{algorithm}
\usepackage{algpseudocode}
\usepackage{amsmath,amsfonts,amssymb,amsthm}
\usepackage{tikz}
\usetikzlibrary{shapes.geometric, arrows.meta, positioning, calc, shadows}
\usepackage{enumitem}
\usepackage{dsfont}
\usepackage{bm}
\usepackage{xcolor}
\usepackage{hyperref}
\usepackage{footmisc}
\usepackage{esvect}
\usepackage{stmaryrd}
\definecolor{burgundy}{rgb}{0.5, 0.0, 0.13}
\definecolor{camel}{rgb}{0.76, 0.6, 0.42}
\definecolor{chamoisee}{rgb}{0.63, 0.47, 0.35}
\definecolor{grey1}{RGB}{128,128,128}
\hypersetup{colorlinks=true,linkcolor=burgundy,citecolor=chamoisee,urlcolor=burgundy,linktoc=page}
\usepackage{natbib}

\newtheorem{theorem}{Theorem}
\newtheorem{lemma}[theorem]{Lemma}

\newtheorem{remark}{Remark} 
\newtheorem{corollary}[theorem]{Corollary}
\newtheorem{definition}[theorem]{Definition}

\newcommand\Z{{\mathds Z}}
\newcommand\N{{\mathds N}}
\newcommand\R{{\mathds R}}
\newcommand\C{{\mathds C}}

\newcommand\E{{\mathds E}}

\newcommand{\Cov}{\mathds{C}\!\mathrm{ov}}
\newcommand{\Var}{{\mathds{V}\!\mathrm{ar}}}

\renewcommand{\k}{\boldsymbol k}

\newcommand{\UN}{\mathbf{1}}

\renewcommand{\to}{\longrightarrow}

\def\cC{\mathcal{C}}
\def\cS{\mathcal{S}}
\def\cH{\mathcal{H}}
\def\cT{\mathcal{T}}

\def\cF{\mathcal{F}}
\def\cE{\mathcal{E}}
\def\diag{\mathrm{Diag}}

\def \d { {\rm d} }
\newcommand{\I}{{\rm Id}}

\newcommand{\EPalm}{\mathds E_{\mathrm{Palm}}}
\newcommand{\PPalm}{\mathds P_{\mathrm{Palm}}}
\newcommand{\EnPalm}{\mathds E_{\mathrm{nPalm}}}
\newcommand{\PnPalm}{\mathds P_{\mathrm{nPalm}}}
\newcommand{\lev}{{ \mathcal L}}
\definecolor{jma-blue}{HTML}{AA66EE}

\title{Gaussian random field's anisotropy using excursion sets}
\usepackage{authblk}
\author[1]{Jean-Marc Aza\"{i}s}
\author[2]{Federico Dalmao}
\author[3,4]{Yohann De Castro}
\affil[1]{{\small Institut de Math\'ematiques de Toulouse\\ Universit\'e  de Toulouse, France.}}
\affil[2]{{\small DMEL, CENUR Litoral Norte\\  Universidad de la República, Salto, Uruguay.}}
\affil[3]{{\small Institut Camille Jordan, CNRS UMR 5208\\  École Centrale Lyon, France.}}
\affil[4]{{\small Institut Universitaire de France~(IUF)}}

\providecommand{\keywords}[1]{\textbf{{Keywords}} #1}

\begin{document}
\maketitle

\begin{abstract}
    This paper addresses the problem of detecting and estimating the anisotropy of a stationary real-valued random field from a single realization of one of its excursion sets. This setting is challenging as it relies on observing a binary image without prior knowledge of the field's mean, variance, or the specific threshold value. 

    Our first contribution is to propose a generalization of Cabaña’s contour method to arbitrary dimensions by analyzing the Palm distribution of normal vectors along the excursion set boundaries. We demonstrate that the anisotropy parameters can be recovered by solving a smooth and strongly convex optimization problem involving the eigenvalues of the empirical covariance matrix of these normal vectors.

    Our second main contribution is a new, model-agnostic statistical test for isotropy in dimension two. We introduce a statistic based on the contour method which is asymptotically distributed as a chi-squared variable with two degrees of freedom under the null hypothesis of quasi-isotropy. Unlike existing methods based on Lipschitz-Killing curvatures, this procedure does not require knowledge of the random field's covariance structure. 
    
    Extensive numerical experiments show that our test is well-calibrated and more powerful than model-based alternatives as well as that the estimation of the anisotropy parameters, including the directions, is robust and efficient. Finally, we apply this framework to test the quasi-isotropy of the Cosmic Microwave Background (CMB) using the Planck data release $3$ mission. 
\end{abstract}


\section{Introduction}

\subsection{Testing anisotropy and estimating its parameters}
Random fields are a standard modelling tool in spatial data analysis and geostatistics with applications in astronomy, hydrogeology, meteorology, oceanography, geochemistry, environmental control, landscape ecology, and agriculture, see for instance \citep{cressie1991statistics,weller2016review,marinucci}. In standard studies on these applications, the random field is often assumed to be stationary and isotropic, the latter meaning that its statistical properties are invariant under translations and rotations. However, in many real-world scenarios, the random field may exhibit anisotropy and/or the hypothesis of isotropy should be tested, with \emph{directions of anisotropy}\footnote{We use the term \emph{directions of anisotropy} to describe the eigenvectors of the covariance of the gradient of the random field. As we will see later, this covariance is used to define anisotropy and its eigenvalues are the \emph{anisotropy parameters}.} unknown to the practitioners. In some scenarios, anisotropy can arise from factors such as geological structures, climatic conditions, or measurement errors. In other scenarios, especially in astronomy, physics models predict that some observed fields should be isotropic and it is of interest to test this hypothesis under anisotropic~alternatives. 

Observations can be costly and/or censored by design so that only a single binary (black and white) image is available, such as in remote sensing, medical imaging, and environmental monitoring. In real applications, the value of the threshold is unknown along with the mean and variance of the random field, which makes the problem more challenging. Hence we consider the question: \emph{Can we estimate the anisotropy of a stationary real valued random field given a single realization of a excursion set associated to some threshold without prior knowledge of its mean, variance and threshold value?}

\newcommand{\widthpicture}{0.24}

\begin{figure}[t]
    \centering
    \includegraphics[width=\widthpicture\linewidth]{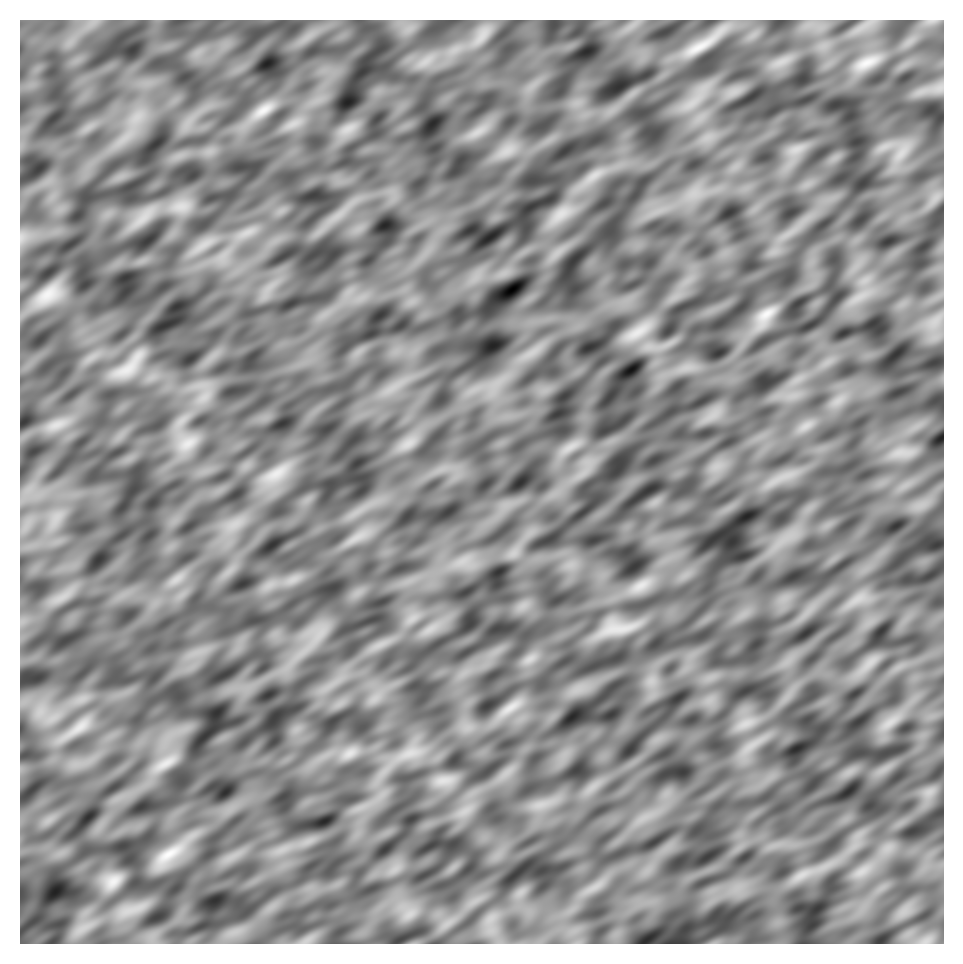}
    \includegraphics[width=\widthpicture\linewidth]{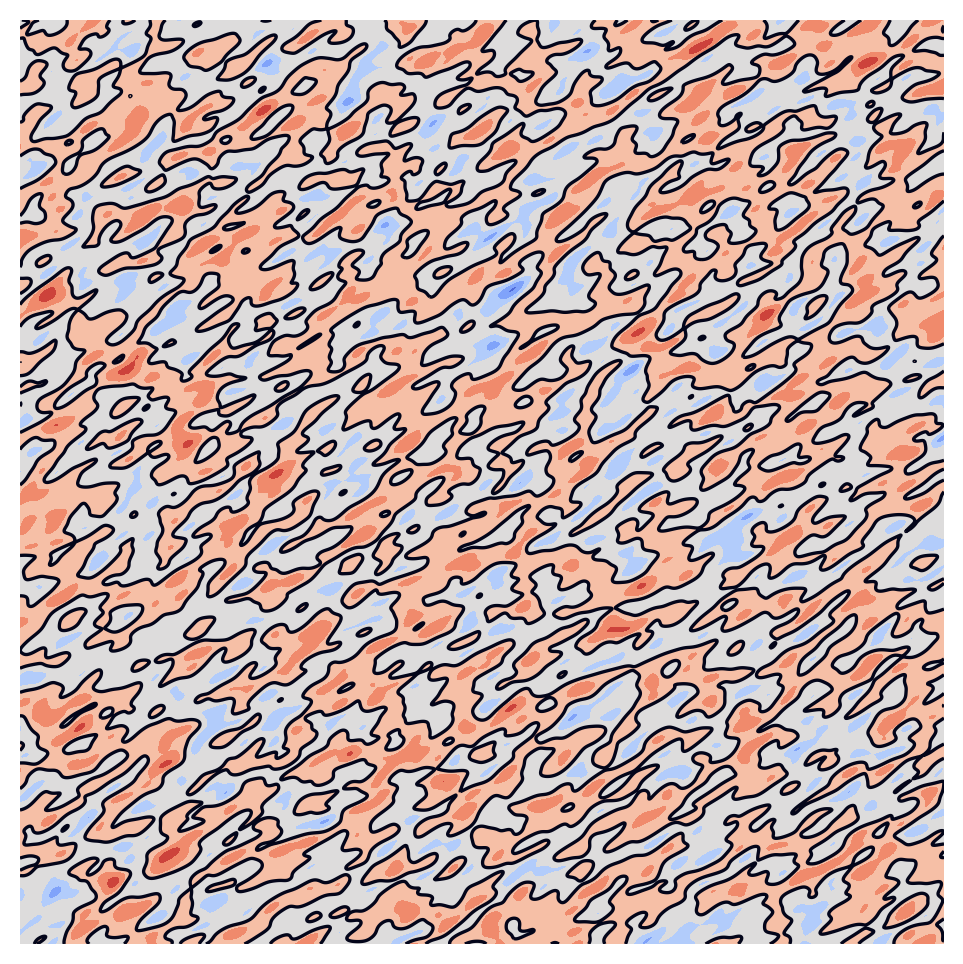}
    \includegraphics[width=\widthpicture\linewidth]{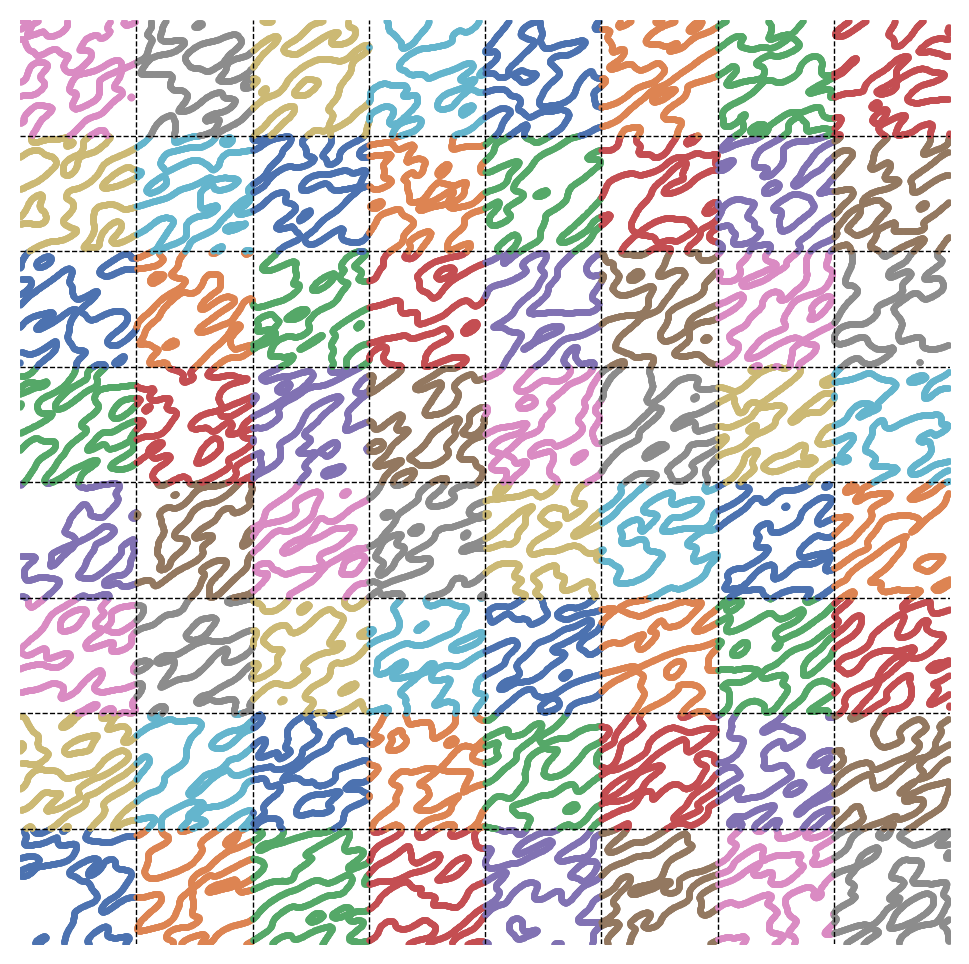}
    \includegraphics[width=\widthpicture\linewidth]{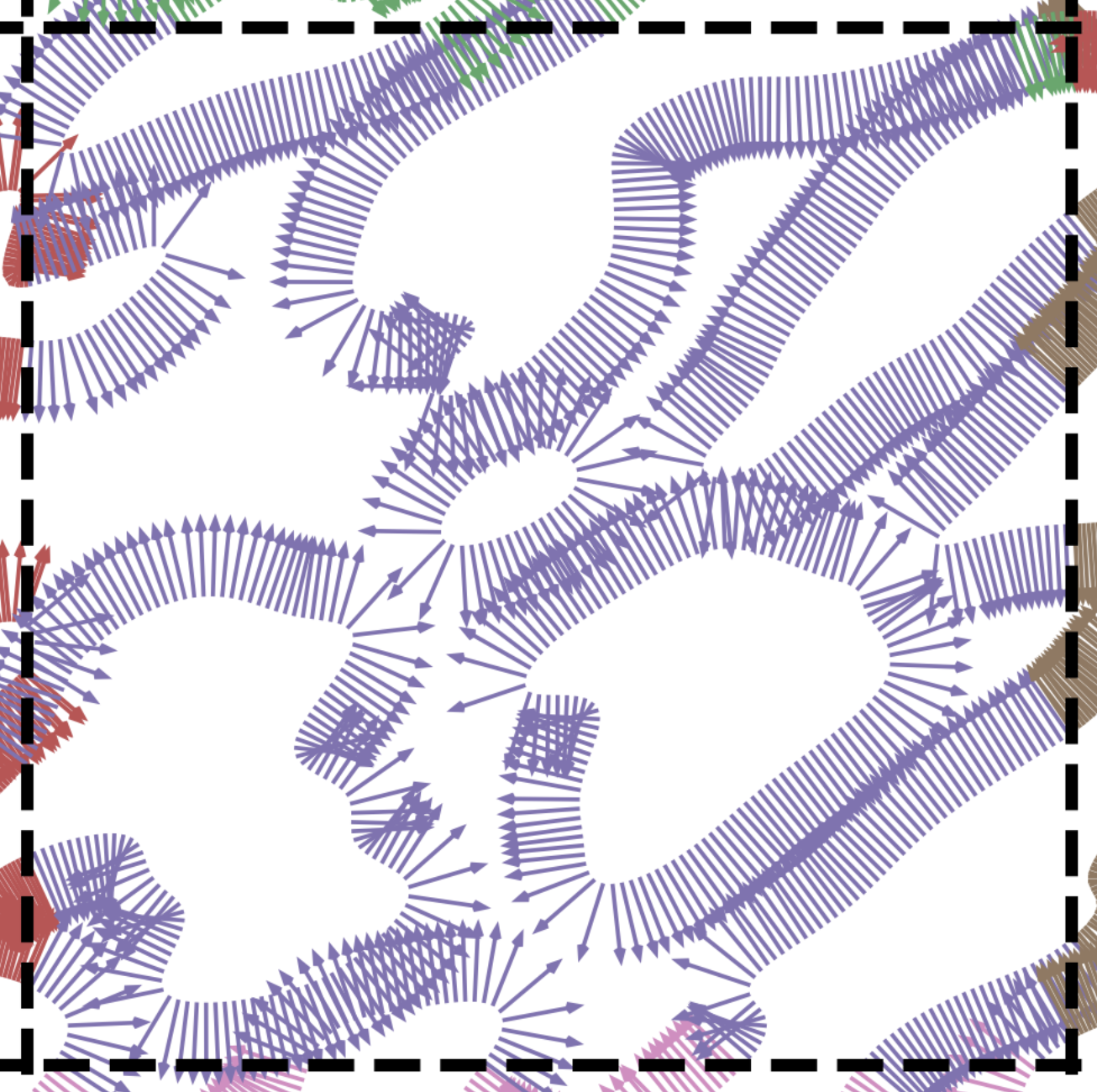}
            \caption{\underline{Left:} $1,000\times 1,000$ image of an anisotropic Gaussian random field ($\kappa=0.9$, $\theta_0=1$), generated with {\tt gstools} \citep{muller2017}. 
            \underline{Middle left:} estimated level set $(X(t)=0)$ by polygons with vertices outside the ($1,000\times 1,000$) grid and optimized using the random field values on the grid only (we used the simple and standard {\tt contour} function of {\tt matplotlib}). \underline{Middle right:} a $8\times 8$ partition of contours on which we have distributed $10^7$ equispaced points with respect to the curvilinear abscissa. 
            \underline{Right:} we consider the $10^7$ normals to the contour points, zoom on the cell $[2,7]$ with its normals.}
    \label{fig:cov_examplebis}
\end{figure}

In order to estimate the anisotropy of the underlying random field, a first possibility is to resort to topological statistics such as the \emph{Euler characteristic} and the \emph{Lipschitz-Killing curvatures} of the excursion set as done in the literature, see \citep{bierme} and references therein. 
This method will be referred to as LKC-method. 

\noindent
As a second possibility, one can also remark  that giving the excursion set is equivalent to giving the level~set, which is its boundary\footnote{At first sight, it may seem that we lose the sign information by considering the level set. But in fact this doesn't matter because the distribution of $X(\cdot)$ and $-X(\cdot)$ coincide. \label{foot:symmetry_X}}. On the level set we can observe the direction of the normal (or the gradient) and study its distribution  
as in the  pioneering works \citep{cabana1987affine} and \citep{wschebor}. This method will be called the contour method. 

The contour method was originally done in dimension two by \citet{cabana1987affine}. Caba\~na's approach, though not yet widely adopted, introduced the key idea of analysing the angular distribution of the normals to level curves. This is accomplished by computing integrals of trigonometric functions of \emph{twice} the normal vector angle, a technique we generalize to any dimensions and build upon. Roughly speaking, the cosine (or sine) of \emph{twice} the normal vector angle is related to the squares and products of cosine and sine and hence to second order statistics such as the covariance of the normal vector angle. This latter can be explicitly described using the \emph{Palm distribution} of the normal vector. Interestingly, this distribution depends only on the anisotropy parameters and directions. As a by-product, \citet{cabana1987affine} and \citet{wschebor} were able to \emph{estimate the anisotropy angle} which is not reachable by the methods using LKC since their distributions are invariant by rotations.
At the time the method was introduced, studying its numerical  properties was out of reach of computational tools. This is perhaps why this method was a little forgotten. 

Our first generalization enables the estimation of anisotropy directions and parameters in any dimension and is, by design, agnostic to the threshold value, mean, and variance of the random field. We show that the directions of anisotropy are given by the ones of the covariance matrix of the normal vector along any level set. We derive the density (with respect to the Lebesgue measure) of the normal vector along any level set, and we refer to it as the Palm density of the normal. More deeply, we prove that the anisotropy parameters are a smooth and strongly concave function of the eigenvalues of this covariance matrix, which can be efficiently inverted by plain gradient descent (GD) with a constant step. This GD enjoys global linear convergence with an explicit rate. 

\pagebreak[3]

Our second generalization revisits the contour method in dimension two to derive testing procedures. Using some Hermite expansion, \cite{berzin} proved that a class of integral functionals, including the Cabaña statistic, are asymptotically\footnote{by asymptotically, we mean that the domain grows to infinity, this notion will be precisely given later in the paper.} normal with a variance which depends on the law of $X(\cdot)$ in a non-tractable manner. We propose to estimate the variance  and introduce a new statistic for the contour method referred to as  \emph{$\chi^2(2)$-Contour}. We show that the limiting law of this new statistic is a $\chi^2$ with two degrees of freedom. The $\chi^2(2)\textrm{-Contour}$ is based on a partition of the domain into cells and we use the normal vectors of each cell to estimate the asymptotic variance, see Figure~\ref{fig:cov_examplebis} for an illustration of the partition and the normal vectors to the contours. This result allows us to build a new test which is \emph{agnostic} to the random field law. A contrario, the standard contour method and the LKC method are model-based (MB) and deriving a testing procedure requires sampling under the model, say Gaussian with a \emph{known} covariance function for instance. We denote these tests by $\textrm{MB-Contour}$ and $\textrm{MB-LKC}$. Model-based tests (MB) have limited applications since they require to know the model under the null beforehand, while the $\chi^2(2)\textrm{-Contour}$ can be used without prior knowledge.

Our results are presented in a Gaussian framework but they can be directly generalised to the case of a \emph{transformed Gaussian random field}: $Y(t) = f (X(t))$ where $X(\cdot)$ is a Gaussian field and $f(\cdot)$ is some strictly monotonic $\mathcal C^2$-function. To summarize, our main contributions are:
\begin{itemize}
\item \emph{A Palm-distribution viewpoint of the contour method:} we derive the Palm law of the (normalized) gradient along level sets and show it depends only on anisotropy, not on the unknown level, mean, variance, or observation window, thereby explaining and extending Caba\~na’s statistics, see Theorem~\ref{thm:Palm_normal}.
\item \emph{A dimension-free contour methodology:} in any dimension, the covariance of Palm-normalized normals has eigenvectors equal to the principal directions and eigenvalues that are smooth functions of the anisotropy parameters (proven to be the gradient of a strictly convex function, hence invertible). Directions are estimated by eigen-decomposition and parameters are recovered by solving a smooth and strictly convex program with guaranteed linear convergence, see Theorem~\ref{thm:inverse_palm}.
\item \emph{Asymptotic theory:} Under mild Arcones-type regularity conditions, we establish multivariate central limit theorems for contour integrals and, via the Delta method, asymptotic normality of the resulting estimators in any dimension.
\item \emph{A model-agnostic isotropy test in 2D:} we construct a chi-squared statistic from cell-wise contour fluctuations with a consistent variance estimator. Under quasi-isotropy the limit law is chi-squared with two degrees of freedom without knowing the field’s covariance.
\item \emph{Numerical study:} extensive experiments compare our approach with LKC-based procedures \citep{bierme} and model-based contour tests; our test is well calibrated and often more powerful, and estimators remain accurate without knowing the level, mean, or variance.
\item \emph{Application to Planck CMB data:} applying our test to the Planck DR3 CMB temperature map \citep{ESA:CMBMaps:2018} rejects quasi-isotropy and recovers the principal direction.
\end{itemize}
\noindent
The key aspects of our methodology are its agnosticism to the field’s mean, variance, and threshold, its applicability in any dimension, and its solid theoretical foundation via Palm distributions and asymptotic normality.

\subsection{Outline}

The paper is organized as follows. Section~\ref{s:framework} introduces the observation model, anisotropy parameters and nuisance quantities, and fixes notation. Section~\ref{s:cont} reviews contour-based methods (Caba\~na and variants), and Section~\ref{s:lkc} recalls LKC-based approaches and their links to anisotropy. Section~\ref{s:main} gathers our main contributions. Section~\ref{sec:palm} develops a Palm-distribution viewpoint of normals along level sets and recovers/extends Caba\~na’s statistics. Section~\ref{s:cont:d} generalizes the contour methodology to any dimension and formulates the convex inversion from Palm normalized-gradient eigenvalues to model eigenvalues. Section~\ref{s:coco} establishes multivariate CLTs for contour integrals and derives the asymptotic normality of the ensuing estimators. Section~\ref{s:chi2} introduces the model-agnostic chi-squared isotropy test in 2D. Section~\ref{s:num} reports numerical experiments: implementation details (level-set extraction and normals), angle recovery, comparison of Contour vs. LKC and against a full-observation oracle, power studies, and an application to Planck CMB data, followed by a short discussion. Appendices provide complementary material: affine-process alternatives (Appendix~\ref{s:affine}), Bulinskaya-type lemmas (Appendix~\ref{s:buli}), elliptic-integral calculations and the explicit form/monotonicity of $g(\kappa)$ (Appendix~\ref{s:palm:calcul}), a simple combination of Contour and LKC estimators (Appendix~\ref{s:combi}), and a table of notation (Appendix~\ref{sec:table_of_notation}).

\subsection{Notations, assumption and framework} \label{s:framework}

To unify presentations of the different papers \citep{cabana1987affine,wschebor,bierme,berzin}, we must choose a common framework and a common notation system. We denote by $X(\cdot) := \{X(t)\,,\ t\in \R^d\}$ a Gaussian stationary random field with values in $\R$ and covariance function $r\,:\,\R^d\to\R$. We set $\mu:= \E (X(0)) \in\R$ and $ \sigma:= \Var (X(0)) \geq0$. In the sections dealing with computation of expectations, we use a parameter domain $\cT$  which is either a bounded open subset of $\R^d$ or the whole $\R^d$. In the sections dealing with asymptotic distributions we use the following family of growing domains: $\cT_n:=(-n,n)^d\,, \ n\in\N$. The notation $\cH^k$ denotes the Hausdorff measure of dimension $k$, $|\cT|:= \cH^d (\cT)$, $|\lev(\cT)|:= \cH^{d-1}(\lev(\cT))$. The notation $\phi(\cdot)$ denotes the standard Gaussian density function, $\Phi(\cdot)$ its cumulative distribution function. A table of notation is given in the Appendix \ref{sec:table_of_notation}.

\noindent
Our results hold under the following standard assumption: 
\begin{equation}    \label{eq:assumption}
    \tag{$\mathds{A}_X$}
    X(\cdot)\text{ is a stationary Gaussian real-valued random field on } \R^d \text{ with }\mathcal C^2 \text{-paths,}\atop \text{satisfying } \sigma^2 := \Var( X(0)) >0 \text{ and }\Lambda := \Var (X'(0)) \text{ does not degenerate.}
\end{equation}

\paragraph{Excursion and level sets.}
We are given a unique realization of an excursion set of the random field~$X(\cdot) $. In other words, for some \emph{unknown} $u\in\R$, we observe the \emph{excursion set}
\begin{subequations}
\begin{equation} \label{e:exc}
 \cE (\cT) := \big\{t \in \cT : X(t) >u\big\}\,.
\end{equation}
Note that this corresponds to observe a black and white image. Apparently, this is different from observing the \emph{level set}
\begin{equation} \label{e:lev}
     \lev(\cT) := \big\{ t\in \cT : X(t) =u\big\}\,.
\end{equation}
\end{subequations}
But, as we have already seen (Footnote~\footref{foot:symmetry_X}), using the fact that $X(\cdot)$ and  $-X(\cdot)$ have the same distribution, we see that the information is exactly the same.

\paragraph{Anisotropy's directions and parameters.}
We denote by $\Lambda $ the variance-covariance matrix of the derivative $X'(0)$.  By diagonalization we get:
\begin{equation} \label{e:var:d}
    \Lambda  := \Var (X'(0))= {\bm P}^\top \ \diag(\kappa^2_1,\ldots, \kappa^2_d) \  {\bm P}\,,
\end{equation}
where ${\bm P}$ is a $d\times d$ unitary matrix that depends on $(d-1)!$ parameters and $\diag$ is a diagonal matrix. In addition we use the convention $\kappa_1^2 \geq \cdots \geq \kappa_d^2 $ and under~\eqref{eq:assumption}, one has  $\kappa_d^2 >0$.

Without loss of generality, we normalize the eigenvalues up to a common scale. Indeed, anisotropy is identifiable only up to a global multiplicative factor: replacing $X(\cdot)$ by $cX(\cdot)$ multiplies $X'(t)$ and $\Lambda$ by~$c$ and $c^2$, respectively, while, after rescaling the (unknown) level $u$ to $cu$, the level and excursion sets are unchanged. We therefore enforce the following normalization:  
\[
\vv{\kappa}:=(\kappa_1,\ldots,\kappa_d)\in\Delta_+:=\Big\{
\vv{\kappa}\in\mathds R^d
\ :\ 
\sum_{i=1}^d\kappa_i^2 =1\,,\ \kappa_i>0\,,\ i=1,\ldots,d
\Big\}
\]
The rows of ${\bm P}$ are called the \emph{directions} of anisotropy and the (square root of the) eigenvalues $\vv{\kappa}$ are called the \emph{parameters of anisotropy}. As long as we restrict our analysis to the computation of expectations at one point $t\in\lev(\cT)$, the parameters of the model are: the expectation $ \mu$, the variance $\sigma^2$ of $X(t)$, $\vv{\kappa}$ and~${\bm P}$. Note that since the observation is given, it is not possible to use a scaling in order to assume, for example, that $\mu=0$ or $\sigma^2=1$. We refer to $(u,\mu,\sigma)$ as \emph{nuisance} parameters since they are unknown in practice.

\paragraph{Quasi-isotropy.} Our aim is to test \emph{quasi-isotropy} as well as estimating some of the parameters of  \eqref{e:var:d}.
The quasi-isotropy is defined by 
\begin{equation} \label{e:h0}
        \Lambda  = c  \  \I_d \mbox{ for some constant }c\,,
       \mbox{ or equivalently  }\kappa^2_1 = \kappa^2_d\,,
\end{equation}
where $\Lambda$ is defined in \eqref{e:var:d} and $\I_d$ stands for the identity matrix of size $d$. 
This assumption is weaker than  the strict isotropy  which demands the distribution of the whole random field to be invariant by isometries while \eqref{e:h0} demands only the distribution of $X'(0)$ to be so. Nevertheless, the computation of the expectations of the quantities we will consider here are equal under both hypotheses. In summary,~\eqref{e:var:d} defines the alternative of our test procedure while the null hypothesis of quasi-isotropy  is defined by~\eqref{e:h0}.

\medskip

\begin{remark}[Quasi-isotropy is strictly larger than isotropy]
    A class of random fields which are quasi-isotropic and not strictly isotropic is the following: Let $\rho(\cdot)$ be a valid covariance of a random process with $\mathcal C^2$-paths, defined 
on $\R$, then 
\[
r(t_1,\ldots,t_d)  = \rho(t_1)\times \cdots \times  \rho(t_d) 
\]
defines a quasi-isotropic covariance on $\R^d$. Excepting the case $\rho(t) = \exp(-at^2)$, it is not isotropic.
\end{remark}

\begin{remark}[Regularity of the field]
In this paper we assume that $X(\cdot)$ is $\mathcal C^2$. This  condition is really  needed  only for LKC-method and the asymptotic results of  Section \ref{s:coco}.
 In other parts  the assumption of~$X(\cdot)$ being  $\mathcal C^1$ is sufficient. We have omitted this point for simplicity.
\end{remark}

\paragraph{Special case $d=2$} In this particular case, \eqref{e:var:d} takes the form:
\begin{subequations}
\begin{equation} \label{e:var}
    \Lambda = {\bm P}_{-\theta_0}\left(\begin{array}{cc}\kappa^2_1& 0 \\0 & \kappa^2_2\end{array}\right) {\bm P}_{\theta_0},
\end{equation}
where 
 \[
  {\bm P}_{\theta_0} = \left(\begin{array}{rr}\cos(\theta_0)& \sin(\theta_0) \\-\sin(\theta_0) & \cos(\theta_0)\end{array}\right), 
 \]
is the rotation with angle $\theta_0$. Additionally the anisotropy can be measured by a single parameter 
\begin{equation}   \label{e:kappa}
  \kappa = \sqrt{ 1-\frac{\kappa_2^2}{\kappa_1^2}}\in[0,1]\,,
\end{equation}
\end{subequations}
with a slight abuse of notation. 

\paragraph{Assumptions and transformed Gaussian random field} Our results are presented in a Gaussian framework but they can be directly generalised to the case of a \emph{transformed Gaussian random field}: $Y(t) = f (X(t))$ where $X(\cdot)$ is a Gaussian field and $f(\cdot)$ is some strictly monotonic $\mathcal C^2$-function. Indeed, the excursion set of $Y(\cdot)$ above level $u$ coincides with the excursion set of $X(\cdot)$ above level $f^{-1}(u)$. As for the variance-covariance matrix of the gradient of $Y(\cdot)$, it is equal to $c_f \Lambda$ where $\Lambda$ is defined in~\eqref{e:var:d} and $c_f=\Var({f'(X(0))})>0$, by independence between $X(0)$ and $X'(0)$ ($X(\cdot)$ is stationary). Since $f'(\cdot)$ is non zero, the anisotropy parameters and directions of $Y(\cdot)$ coincide with those of $X(\cdot)$. For sake of clarity, we have chosen to present our results in the Gaussian framework only, which corresponds to the case $f(x)=x$ assumed here without loss of generality. 

\section{Existing methods}
As depicted in Figure~\ref{fib:flow_chart_Contour_LKC} we present two estimation strategies, one based on orientation of the normal vectors (Contour method) and one based on topology (LKC method). All the notation of Figure~\ref{fib:flow_chart_Contour_LKC} will be introduced along the next sub-sections, this diagram is meant to describe the ‘‘pipelines'' involved in estimations.

 \begin{figure}[h]
\centering
\resizebox{\textwidth}{!} { 
\begin{tikzpicture}[
    node distance=0.6cm and 0.8cm,
    base/.style={rectangle, rounded corners, draw=black, drop shadow,
                 text centered, font=\small\sffamily, inner sep=4pt},
    startstop/.style={base, fill=gray!20, minimum height=0.8cm},
    process/.style={base, fill=white, text width=2cm, minimum height=1.1cm},
    decision/.style={base, fill=blue!5, text width=1.8cm, minimum height=1cm},
    arrow/.style={thick, -Latex, color=black!80},
    contour/.style={fill=blue!10, draw=blue!50!black},
    lkc/.style={fill=red!10, draw=red!50!black}
]

    \node (input) [startstop, text width=2cm] {Input $X(t)$};
    \node (excursion) [process, right=of input, fill=yellow!10, text width=2.5cm] {Exc. Set $\mathcal{E}(\mathcal{T})$\\$\{X(t) \geq u\}$};

    
    \node (normals) [process, contour, above right=-0.2cm and 1cm of excursion] {Normals\\$N(t)$};
    \node (trig) [process, contour, right=of normals] {Trig. Integrals\\$\mathcal{C}, \mathcal{S}$};
    \node (invG) [decision, contour, right=of trig] {Invert\\$g(\kappa)$};
    \node (outC) [startstop, contour, right=of invG, text width=1.5cm] {$\widehat{\kappa}_C, \widehat{\theta}_0$};

    \node (euler) [process, lkc, below right=-0.2cm and 1cm of excursion] {Euler Char.\\$GC$};
    \node (invR) [decision, lkc, right=of euler, xshift=2.9cm] {Invert\\$R(\kappa)$}; 
    \node (outLKC) [startstop, lkc, right=of invR, text width=1.5cm] {$\widehat{\kappa}_{\rm LKC}$};

    \draw[arrow] (input) -- (excursion);
    
    \draw[arrow] (excursion.east) -- ++(0.5,0) |- node[near start, above, font=\bfseries] {} (normals.west);
    \draw[arrow] (excursion.east) -- ++(0.5,0) |- node[near start, below, font=\bfseries] {} (euler.west);

    \draw[arrow] (normals) -- (trig);
    \draw[arrow] (trig) -- (invG);
    \draw[arrow] (invG) -- (outC);

    \draw[arrow] (euler) -- (invR);
    \draw[arrow] (invR) -- (outLKC);

    \draw[dashed, gray!50] ($(normals.north west)+(-0.2,0.2)$) rectangle ($(outC.south east)+(0.2,-0.4)$);
    \node[anchor=south east, font=\itshape, color=blue!60!black] at ($(normals.north west)+(4.5,0.1)$) {{\bf Contour} $($Orientation-based$)$};

    \draw[dashed, gray!50] ($(euler.north west)+(-0.2,0.2)$) rectangle ($(outLKC.south east)+(0.2,-0.4)$);
    \node[anchor=south east, font=\itshape, color=red!60!black] at ($(euler.north west)+(3.5,-1.9)$) {{\bf LKC} $($Topology-based$)$};

\end{tikzpicture}
 } 
\caption{Comparison of estimation pipelines. Both methods rely on the excursion set, but diverge in geometric summary statistics: the Contour method uses boundary normals (orientation), while LKC uses curvature and Euler characteristics (topology).}
\label{fib:flow_chart_Contour_LKC}
\end{figure}
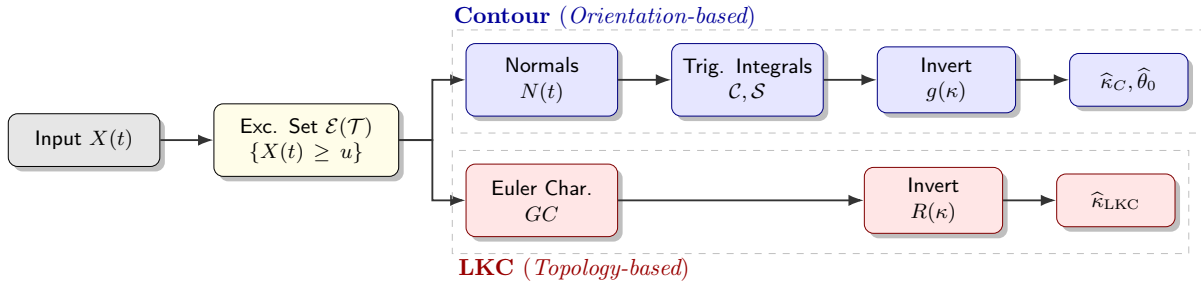
 
\subsection{Contour methods} 
\label{s:cont}
Originally, \citet{cabana1987affine} and \citet{wschebor} used non-Gaussian models. This difference is detailed in Section \ref{s:affine}. The method of Caba\~na is based on the idea that, in case of anisotropy, the distribution of the direction of the gradient along the level curve is far from being uniform on the sphere $\mathds{S}^{d-1}$. The paper considers computations of three integrals along the level curve: 
\begin{align*}
    |\lev(\cT)|&:= \cH^1( \lev(\cT) ) &\text{the length of the level curve,}\\
    \cC(\cT)&:= \int_{\lev(\cT)}  \cos\big( 2\Theta (t) \big)  \d \cH^1(t) &\text{the cosine integral,}\\
    \cS(\cT)&:= \int_{\lev(\cT)}  \sin\big( 2\Theta (t) \big)
 \d \cH^1(t)&\text{the sine integral,}
\end{align*}
where $\Theta(t) $ is the angle of the gradient  at point $t$. Note that Lemma \ref{l:buli1} shows that a.s. $\Theta(t)$ is well defined almost everywhere on the level curve. Moreover, it is observable from $\cE(\cT)$ or $\lev(\cT)$. The paper \citep{cabana1987affine} provides no precise explanation for this choice. We give an interpretation in Section \ref{sec:palm}.
 
\paragraph{Nuisance-free normalized cosine and sine.} By the use of three Kac-Rice formulas, Theorems 2.2 and 6.1 of \citep{armentano2025general} without any further assumptions than~\eqref{eq:assumption}:
\begin{subequations}
    \begin{align} 
    \label{e:peri}
    \E\big( |\lev(\cT)| \big) &= |\cT|\, \E(\|X'(0) \|)\  \frac1\sigma\phi\Big(\frac{u-\mu}\sigma\Big)\,, \\
    \E(\cC(\cT)) &= |\cT|\, \E\big(\|X'(0) \| \cos(2\Theta(0))  \frac1\sigma\phi\Big(\frac{u-\mu}\sigma\Big)\,, \\
 \E(\cS(\cT))& = |\cT|\, \E\big(\|X'(0) \| \sin(2\Theta(0))\frac1\sigma\phi\Big(\frac{u-\mu}\sigma\Big)\,.
    \end{align}
\end{subequations}
We see clearly that computing the quotients $\frac{ \E(\cC(\cT))}{ \E( |\lev(\cT)|) }$ and $\frac{ \E(\cS(\cT))}{ \E( |\lev(\cT)|) }$ permits to get rid of the nuisance parameters $u$, $\mu$, $\sigma$, and also of $|\cT|$. A calculation that will be detailed in Section \ref{s:palm:calcul} gives
\begin{subequations}
\begin{equation}   \label{e:cabana}
  \frac{\E(\cC(\cT)) }{\E(|\lev(\cT)|)} = \cos(2\theta_0)  g(\kappa), \quad   \frac{\E(\cS(\cT)) }{ \E(|\lev(\cT)|)} = \sin(2\theta_0)  g(\kappa),
\end{equation}
where  $\kappa$ has been defined in \eqref{e:kappa} and $g(\cdot)$ is defined by 
 \begin{equation} \label{def:G}
    g(\kappa)
    := \frac
                {\int_{-\pi}^\pi (2\cos^2\!\theta-1)\,\big(1-\kappa^2\cos^2\!\theta\big)^{-\frac{3}{2}}\mathrm d\theta}
                {\int_{-\pi}^\pi \big(1-\kappa^2\cos^2\!\theta\big)^{-\frac{3}{2}}\mathrm d\theta}.\\
\end{equation} 
\end{subequations}
The ratios \ref{e:cabana} are referred to as the \emph{normalized cosine and sine}, they are nuisance-free as they do not depend on the nuisance parameters $(u,\mu,\sigma)$. The paper \citep{cabana1987affine} gives an expression of $g(\cdot)$ that contains some typos and the correct expression of $g(\cdot)$ is given in \ref{def:G}. In Section \ref{s:palm:calcul} we give a derivation of this formula and the following lemma is proven there. A plot of $g(\cdot)$ is given in Figure \ref{fig:inverse_MB}. 

\begin{lemma} \label{l:g}
The function $g(\cdot)$ is invertible.
\end{lemma}

Moreover, the normalized cosine and sine integrals given in \eqref{e:cabana} can be interpreted as the coordinates of a vector in $\R^2$ whose angle is $2\theta_0$ and whose norm is $g(\kappa)$. Their empirical counterparts are studied and represented in Figure \ref{fig:Cabana_MB}.
\paragraph{Anisotropy angle and parameter.}
\begin{subequations}
 Using \eqref{e:cabana}, the anisotropy angle $\theta _0$ is estimated by 
 \begin{equation} \label{e:theta:hat}
     \widehat{\theta}_0 = \frac12 \arctan \Big( \frac{\cS(\cT)}{\cC(\cT)} \Big),
 \end{equation} 
 and since 
 \begin{equation}
  \frac{\E^2(\cC(\cT)) + \E^2(\cS(\cT))  }{ \E^2(|\lev(\cT)|)} = g^2( \kappa),  
 \end{equation}
the anisotropy parameter $\kappa $ is estimated by 
\begin{equation} \label{e:k:hat}
     \widehat{\kappa}_{\rm{C}} = g^{-1} (\cF)
\end{equation}
 with    
 \begin{equation}\label{e:cf}
   \cF  :=  \frac{ \sqrt{\cC(\cT)^2 + \cS(\cT)^2 } }{ |\lev(\cT)|} .  
\end{equation}
\end{subequations}

\paragraph{Anchored estimation.}
To end this section let us mention that the method proposed in \citep[P. 79-85]{wschebor} and developed in \citep{berzin} is a variant of the method by Cabaña. Instead of considering the integrals $\cS(\cT)$ and $\cC(\cT)$, we choose a particular direction~$\theta_1 $ and consider the integral on the level curve of
\[
 \big(\cos(\Theta(t)), \sin(\Theta(t))\big)\; {\rm sign}\big(\cos(\Theta(t)-\theta_1 )\big).
\]
This method is described in detail in \citep{berzin}. It has the disadvantage of being dependent on some extra parameter $\theta_1$ which is difficult to determine. Moreover, the computations are more involved than in the Cabaña's method and the generalization to higher dimension is not easy. 

\subsection{Lipschitz--Killing curvature based estimation} \label{s:lkc}

We now present an anisotropy estimation method built on Lipschitz--Killing curvatures (LKCs) of excursion and level sets, following the curvature density approach of \citet{bierme} (see also the classical treatment in \citep{adler} and the Minkowski tensor formulation in \citep{klatt2022characterization}). 

\paragraph{Lipschitz--Killing curvatures}
Using Lemma \ref{l:buli2}, a.s. the level set $\lev(\R^d)$ contains no points $t$ such that $X'(t) =0$. Using the implicit function theorem, the level set $\lev(\R^d)$ is a.s. a $\mathcal{C}^2$ manifold. In such a case, the $j$-th, $j=0,\ldots,d$,  (global) Lipschitz-Killing curvature is defined by
\[
 C_j := \frac{1}{\alpha_{d-1-j}} \int_{\lev(\R^d)} \sigma_{d-1-j}\big(k_1(t), \ldots, k_{d-1}(t)\big) \, d\cH^{d-1}(t),
\]
where $k_i(t)$ are the principal curvatures (eigenvalues of the second fundamental form), $\sigma_{m}$ the $m$-th elementary symmetric function and $\alpha_{d-1-j}$ the volume of the Euclidean unit ball in dimension $d-1-j$. Restricting to an observation window $\cT$ we obtain the \emph{localized} quantity
\begin{equation} \label{e:gauss}
 C_j(\cT) := \frac{1}{\alpha_{d-1-j}} \int_{\lev(\cT)} \sigma_{d-1-j}\big(k_1(t), \ldots, k_{d-1}(t)\big) \, d\cH^{d-1}(t),
\end{equation}
with $C_j = C_j(\R^d)$.

\paragraph{Expected geometric summaries in 2D.} Denote $w:=(u-\mu)/\sigma$. In $d=2$ the three relevant LKCs of the excursion set $\cE(\cT):=\{t\in\cT : X(t) \ge u\}$ are (up to conventional constants) the area, the boundary length and the Euler characteristic (or equivalently the Gaussian curvature integral, differing only by boundary terms), namely:
\begin{itemize}
    \item Area: $\E\big( \cH^2(\cE(\cT))\big) = |\cT|\, \Phi(-w) = |\cT|\, \Phi\big(\frac{\mu-u}{\sigma}\big)$.
    \item Boundary length (perimeter of the level set): by \eqref{e:peri},
    \[
     \E\big( \cH^1(\lev(\cT)) \big) = |\cT|\, \E\big(\|X'(0)\|\big) \frac{1}{\sigma} \phi(w) = |\cT| \sqrt{\frac{2}{\pi}}\, \kappa_1 E(\kappa) \frac{1}{\sigma} \phi(w),
    \]
    where $E(\kappa)=\int_0^{\pi/2} (1-\kappa^2 \sin^2 \theta)^{1/2} \, d\theta$.
    \item Gaussian curvature (Euler characteristic proxy):
    \[
     \E\big( \mathrm{GC}(\cE(\cT)) \big) = |\cT| \frac{\kappa_1 \kappa_2}{2\pi\sigma^2} w \phi(w).
    \]
\end{itemize}

\paragraph{Estimation procedure (LKC method).} Given one observation over $\cT$:
\begin{subequations}
\begin{enumerate}
    \item Compute empirical area $A:=\cH^2(\cE(\cT))$ and set \( \hat w := -\Phi^{-1}(A/|\cT|) \). If one has only access to level sets and not excursion sets, one can always define a black and white image from it and decides arbitrarily wether excursions are represented by black regions or white regions. By symmetry of the Gaussian random field law, the procedure works for any of these choices.
    \item Measure level set length $L:=\cH^1(\lev(\cT))$ and define
    \begin{equation} \label{e:hatp}
        \widehat P := \frac{1}{|\cT|} \frac{ \phi(\hat w)}{2\sqrt{2\pi}} \, L.
    \end{equation}
    This estimates $(\kappa_1 E(\kappa))/\sigma$.
    \item Compute Gaussian curvature integral $\mathrm{GC}(\cE(\cT))$ and set
    \begin{equation} 
    \label{e:hatec}
        \widehat{\mathrm{GC}} := 2\pi \frac{\mathrm{GC}(\cE(\cT))}{|\cT| \hat w \phi(\hat w)}.
    \end{equation}
    This estimates $(\kappa_1 \kappa_2)/\sigma^2$.
    \item Form the anisotropy ratio\footnote{Their $R(\alpha)$ reads $(\pi^2/4) R(\kappa)$ with $\alpha=1-\kappa^2$ in our notations. A note on their notion of ‘‘Almond curve'' is given in Section~\ref{sec:almond}.}  of \citep{bierme},
    \[
        R(\kappa) := \frac{\kappa_1 \kappa_2}{(\kappa_1 E(\kappa))^2} = \frac{\sqrt{1-\kappa^2}}{(E(\kappa))^2},
    \]
    and estimate it by $\widehat R := \widehat{\mathrm{GC}}/(\widehat P)^2$.
    \item Invert $R$ numerically to obtain
    \begin{equation} \label{e:kappa:LKC}
     \widehat \kappa_{\mathrm{LKC}} := R^{-1}(\widehat R), \quad \text{if } \widehat R \in \Big[0, \frac{4}{\pi^2}\Big].
    \end{equation}
    Outside the admissible range we truncate: set $\widehat \kappa_{\mathrm{LKC}}=1$ if $\widehat R<0$ and $\widehat \kappa_{\mathrm{LKC}}=0$ if $\widehat R>4/\pi^2$ (effectively isotropic), see Figure \ref{fig:inverse_MB}.
\end{enumerate}
\end{subequations}

\paragraph{Remarks.} (i) The ratio $R$ cancels the variance scale $\sigma$, the level threshold $u$ and the mean $\mu$ and therefore is robust to nuisance parameters. (ii) Higher-dimensional generalizations keep the same philosophy, using curvature densities of codimension-one level sets (it is done for $d=3$ in \citep{bierme}). (iii) Related one-dimensional perimeter-based approaches as \citep{estrade} require additional model restrictions.

We retain $\kappa$ itself as primary anisotropy parameter and will compare $\widehat \kappa_{\mathrm{LKC}}$ with the contour-based estimate in Section \ref{s:num}.

\section{Main results} \label{s:main}

\subsection{An interpretation of the Contour method using Palm distributions} \label{sec:palm}

In this section, we give a justification of the intuition hidden in the paper \citep{cabana1987affine} which introduced the contour method and investigated the cosine of twice the angle along the level set. It was not clear from which point this analysis stemmed. Within our framework, in dimension $2$, we consider the \emph{Palm distribution} \citep{coeurjolly2017tutorial} of the gradient $X'(t)$ along the level curve $\lev(\cT)$ and we will uncover and generalize Cabaña's approach. 

The Palm distribution is defined as the distribution of the gradient $X'(t)$ at a point $t$ \emph{chosen  at random} on the level set $\lev(\cT)$. It is computed by means of two Kac-Rice formulas. The first one computes the mean length of the level curve  restricted  to the condition $\{X'(t) \in B\}$, where $B$  is some Borel  set of~$\R^2$. The result is, because of stationary properties, using \citep[Th 7.1]{armentano2025general}.
\begin{subequations} 
\begin{align*}
    \E \big(  \cH^1( \{t\in \lev(\cT) : X'(t) \in B \}) \big) =  
    |\cT|\, \E(\|X'(0) \| \UN_{X'(0) \in B} )\ \frac1\sigma\phi\Big(\frac{u-\mu}\sigma\Big)\,.
\end{align*}
The second Kac-Rice formula gives the mean length of the level curve. Recalling \eqref{e:peri}:
\begin{align*}
    \E\big( |\lev(\cT)| \big) = |\cT|\, \E(\|X'(0) \|)\  \frac1\sigma\phi\Big(\frac{u-\mu}\sigma\Big)\,.
\end{align*}
\end{subequations}
So that the Palm distribution of the gradient is given by
\begin{align}
\label{eq:palm_def}
\forall B\subseteq\R^d\text{ Borel set}\,,\quad 
    \PPalm (B) := \frac{\E(\|X'(0) \| \UN_{X'(t) \in B})} {\E(\|X'(0) \| )}\,.
\end{align}
Furthermore, we can derive its density from Kac-Rice formulas. 

\noindent
{\bf Fact:} {\it The Palm distribution is well defined and it does not depend on the nuisance parameters  $u,\mu, \sigma$ and $\cT$. Suppose, for the moment, that we are in the eigenbasis given by the directions of anisotropy, see~\eqref{e:var}.
The Palm distribution has a density given by
\begin{align*}
C_{\kappa_1,\kappa_2} \sqrt{x_1^2 +x_2^2}\,
\exp\Big\{ -\frac12 \Big( \frac{x_1^2}{\kappa_1^2} +\frac{x_2^2}{\kappa_2^2} \Big) \Big\},
\end{align*}
where $C_{\kappa_1,\kappa_2}$ denotes a normalizing constant depending on $\kappa_1$ and $\kappa_2$.}

To find the distribution of the angle $\Theta$ of the gradient, we switch to polar coordinates. The density becomes 
\[
C'_{\kappa_1,\kappa_2} \rho^2 \exp \left\{ -1/2 \big( (\rho \cos (\theta)/\kappa_1)^2 + (\rho \sin (\theta)/\kappa_2)^2 \big) \right\}.
\]
By integrating  in $\rho$, we obtain the Palm density of the angle $\Theta$:
\[
C''_{\kappa_1,\kappa_2} \big( \cos^2(\theta)/ \kappa^2_1 + \sin^2(\theta)/ \kappa^2_2\big)
^{ -3/2} = 
C'''_{\kappa_1,\kappa_2} \big( 1-\kappa^2\cos^2(\theta) \big)
^{ -3/2} ,
\]
where $\kappa$ has been defined in \eqref{e:kappa} and $C'_{\kappa_1,\kappa_2},C''_{\kappa_1,\kappa_2}$
and $C'''_{\kappa_1,\kappa_2}$ are defined in the same manner as $C_{\kappa_1,\kappa_2}$.
Finally, taking into account the rotation with angle $\theta_0$, we get that the density is actually
\begin{align} \label{e:density}
    C'''_{\kappa_1,\kappa_2} \big( 1 -\kappa^2  \cos^2(\theta- \theta_0)\big)^{ -/3/2}\,.
\end{align}

This analysis can be generalized to any dimension $d\geq2$. Define the \emph{Palm distribution of the gradient along the level set $\lev(\cT)$ at level $u$} as
\begin{equation*}
    \PPalm(B):=\frac{\E\big(\int_{\lev(\cT)}\UN_{X'(t)\in B}\mathrm{d}t\big)}{\E\big(|\lev(\cT)|)}
    \quad\text{for any Borel set }B\in\mathds R^d\,.
\end{equation*}
Using two Kac-Rice formulas \cite[Th. 7.1]{armentano2025general}, one for the numerator and one for the denominator, we obtain the following theorem proving that the Palm law does not depend on the level~$u$ nor on $\cT$.

\begin{theorem}[Palm density of the gradient]
\label{thm:Palm_gradient} With the notation of Section \ref{s:framework} and Assumption~\eqref{eq:assumption}, the Palm  law of the gradient along $\lev(\cT)$ is given for any Borel set $B\subseteq\R^d$ by
\[
\PPalm(B)=\frac{\E\big(\|X'(0)\|\;\UN_{\{X'(0)\in B\}}\big)}{\E\|X'(0)\|}\,.
\]
And, it has a density with respect to the Lebesgue measure on $\R^d$ given by
\begin{equation}
\label{eq:palm_density}
    p_{\mathrm{Palm}}(x)=c_\Lambda\,\|x\|\,(2\pi)^{-d/2}(\det\Lambda)^{-1/2}\,
\exp\!\Big(-\tfrac12 x^\top \Lambda^{-1}x\Big),\qquad x\in\R^d\,,
\end{equation}
where \(c_\Lambda^{-1}=\E\|X'(0)\|\). In particular, $\PPalm$ does not depend on the nuisance parameters $(u,\mu,\sigma)$ and~$\cT$.

\pagebreak[3]
In dimension $d=2$, write \(\Lambda=\mathbf P_{-\theta_0}\,\diag(\kappa_1^2,\kappa_2^2)\,\mathbf P_{\theta_0}\) with
\(\kappa_1\ge \kappa_2>0\) and \(\kappa:=\sqrt{1-\kappa_2^2/\kappa_1^2}\in[0,1)\).
If $\theta$ denotes the angle of $X'$, its Palm density on $(-\pi,\pi]$ is
\[
f_\Theta(\theta)=C_\kappa\big(1-\kappa^2\cos^2(\theta-\theta_0)\big)^{-3/2},
\quad
C_\kappa=\Big(\int_{-\pi}^{\pi}\big(1-\kappa^2\cos^2\varphi\big)^{-3/2}\,d\varphi\Big)^{-1}.
\]
\end{theorem}

Note that, if $\kappa=0$, the latter reduces to the density of the uniform distribution.
Suppose now that~$\kappa $ is small. Then, we have the following Taylor expansion for the density in \eqref{e:density}
\begin{align*}
  \big( 1 - \kappa^2  \cos^2(\theta- \theta_0)\big)^{ -3/2} &\simeq   1 +\frac34 \kappa^2( 1-\cos(2(\theta- \theta_0))) \\
   &= 1+\frac34 \kappa^2 + \kappa^2  \big(\cos(2\theta) \cos(2 \theta_0) -\sin(2\theta)\sin(2\theta_0) \big),
\end{align*}
which shows that \emph{the first departure from uniformity under small anisotropy sits in the second harmonic $\cos(2(\theta-\theta_0))$}. Therefore, integrating $\cos(2\theta)$ and $\sin(2\theta)$ along the level set isolates exactly this harmonic: under the Palm law one gets
$\EPalm[\cos(2\theta)]=\cos(2\theta_0) g(\kappa)$ and $\EPalm[\sin(2\theta)]=\sin(2\theta_0) g(\kappa)$,
with $g(\kappa)$ coming from the normalizing constants. This is precisely the Caba\~na statistic used in this paper: 
\begin{align*}
    \theta_0&=(1/2)\arctan (\EPalm[\sin(2\theta)]/\EPalm[\cos(2\theta)])\\
    \kappa&=g^{-1}(\sqrt{\EPalm[\cos(2\theta)]^2 + \EPalm[\sin(2\theta)]^2}/|\lev(\cT)|)\,.
\end{align*}

\subsection{Contour method in higher dimensions} \label{s:cont:d}
The aim of this section is to extend, using the Palm distribution, the Cabaña method to higher dimensions. First, remark that we can give an equivalent presentation of this last method computing the integrals
\[
\cC'(\cT):= \int_{\lev(\cT)} \!\!\!\!\!\!  \cos^2\big(\Theta (t) \big)  \d \cH^1(t) \quad\text{and}\quad
\cS'(\cT):= \int_{\lev(\cT)} \!\!\!\!\!\!  \cos\big(\Theta (t)\big)\sin\big(\Theta (t) \big)  \d \cH^1(t)\,,
\]
instead of the less intuitive $\cC(\cT)$ and $\cS(\cT)$, the modifications are direct. 
In this form, the method can be generalized to $d>2$. For $t$ belonging to $\lev(\cT)$ define $N(t):=X'(t)/\|X'(t)\|$ as the normalized gradient. By Bulinskaya-type arguments (see Lemma~\ref{l:buli1}), $N(t)$ is well defined $\cH^{d-1}$-a.e. on $\lev(\cT)$. Now, choosing arbitrarily an orientation and computing, 
\begin{equation} \label{e:cij}
 \cC_{i,j}(\cT):= \int_{\lev(\cT)}\!\!\!\!\!\!  N_i(t) N_j(t) \d \cH^{d-1}(t), \quad  i,j = 1,\ldots,d\,,
\end{equation}
one can get from $\cC_{i,j}(\cT)/|\lev(\cT)|$ the empirical variance-covariance matrix associated to the Palm density of the normalized gradient belonging to the unit sphere $\mathds S^{d-1}$. 

\subsubsection*{Variance-covariance matrix of the random normal vector}

Let ${N(t)\in\mathds S^{d-1}}$ be the random unit vector given by the direction of the gradient $X'(t)$ at a point ${t\in\mathds R^d}$ uniformly chosen on the  level set $\lev(\cT)$. We will prove that $N(t)$
follows a probability distribution denoted by $\PnPalm$ standing for the Palm distribution of the \emph{normalized} gradient. Expectation with respect to this distribution will be denoted by $\EnPalm$.
Using the same route as in Section~\ref{sec:palm}, the next theorem proves that the Palm distribution of the normalized gradient reads as
\[
\frac{\mathrm d\PnPalm(\zeta)}{\mathrm d\eta(\zeta)}
= C_{\Lambda}\,\big(\zeta^\top \Lambda^{-1}\zeta\big)^{-\frac{d+1}{2}}
=C_{\Lambda}\|{\bm D}^{-\frac12}\mathbf{z}\big\|^{-(d+1)}
=C_{\Lambda}\Big(\frac{z_1^2}{\kappa_1^2}+\cdots+\frac{z_d^2}{\kappa_d^2}\Big)^{-{\frac{d+1}{2}}},
\]
where $\mathbf{z}={\bm P}\zeta$, $\eta$ is the uniform measure on $\mathds S^{d-1}$ and ${\bm D}:=\diag(\kappa^2_1,\ldots, \kappa^2_d)$ with $\kappa:=(\kappa_1, \ldots,\kappa_d)\in\Delta_+$. Introduce,
\begin{align} 
\label{eq:def_Z}
\mathcal{Z}(\vv\kappa):=(\mathcal{Z}_1(\vv\kappa),\ldots,\mathcal{Z}_d(\vv\kappa))\quad
\mathrm{where}
\quad
\mathcal Z_\ell(\vv\kappa):=
        C_{\Lambda} \int_{\mathds S^{d-1}}z_\ell^2 \,\|{\bm D}^{-\frac12}\mathbf{z}\big\|^{-(d+1)}\mathrm d\eta(\mathbf{z})\,,
\end{align}
referred to as the \emph{Palm normalized gradient eigenvalues}. 

\begin{theorem}[Palm density of the normal]
\label{thm:Palm_normal} With the notation of Section \ref{s:framework} and Assumption~\eqref{eq:assumption}, one can define the Palm law $\PnPalm$ of the normalized gradient along $\lev(\cT)$ by
\[
\int_{\mathds S^{d-1}}\!\!\!\!\!\varphi(\zeta)\,\mathrm d\PnPalm(\zeta)
=\frac{\E\!\big[\int_{\lev(\cT)}\!\varphi\big(N(t)\big)\,\mathrm d\cH^{d-1}(t)\big]}
    {\E\!\big[\cH^{d-1}\big(\lev(\cT)\big)\big]}\,,\qquad
\forall\,\varphi:\mathds S^{d-1}\to\R\text{ bounded measurable}.
\]
Then $\PnPalm$ is absolutely continuous w.r.t. the uniform probability measure $\eta$ on $\mathds S^{d-1}$ with density
\[
\frac{\mathrm d\PnPalm(\zeta)}{\mathrm d\eta(\zeta)}
= C_{\Lambda}\,\big(\zeta^\top \Lambda^{-1}\zeta\big)^{-\frac{d+1}{2}}
= C_{\Lambda}\,\Big(\frac{z_1^2}{\kappa_1^2}+\cdots+\frac{z_d^2}{\kappa_d^2}\Big)^{-\frac{d+1}{2}},
\quad \mathbf z:={\bm P}\zeta,
\]
where the normalizing constant is
\[
C_{\Lambda}^{-1}=\int_{\mathds S^{d-1}}\,\big(\zeta^\top \Lambda^{-1}\zeta\big)^{-\frac{d+1}{2}}\,\mathrm d\eta(\zeta).
\]
In particular, $\PnPalm$ does not depend on the nuisance parameters $(u,\mu,\sigma)$ and $\cT$, and it holds that 
\begin{equation} \label{eq:covariance_Palm}
N(t)\sim \PnPalm\quad\text{and}\quad
\E(N(t)N(t)^\top)=
\EnPalm\big(\zeta\zeta^\top\big)=
{\bm P}^\top \diag(\mathcal Z(\vv\kappa)){\bm P}\,.
\end{equation}
\end{theorem}
\begin{proof}
By Theorem~\ref{thm:Palm_gradient}, the Palm law of the gradient $X'(t)$ has Lebesgue density
\[
p_{\mathrm{Palm}}(x)=c_\Lambda\,\|x\|\,(2\pi)^{-d/2}(\det\Lambda)^{-1/2}
\exp\!\Big(-\tfrac12 x^\top \Lambda^{-1}x\Big),\qquad x\in\R^d.
\]
For any bounded measurable $\varphi$ on $\mathds S^{d-1}$, the Palm law of the normalized gradient $N(t)=X'(t)/\|X'(t)\|$ is the push-forward of $\PPalm$ by $x\mapsto x/\|x\|$, hence
\[
\int_{\mathds S^{d-1}} \!\varphi(\zeta)\,\mathrm d\PnPalm(\zeta)
= \int_{\R^d} \!\varphi\!\left(\frac{x}{\|x\|}\right) p_{\mathrm{Palm}}(x)\,\mathrm dx.
\]
Switching to spherical coordinates $x=r\,\zeta$ with $r>0$ and $\zeta\in\mathds S^{d-1}$ gives $\mathrm dx=r^{d-1}\,\mathrm dr\,\mathrm d\eta(\zeta)$ and
$x^\top\Lambda^{-1}x=r^2\,\zeta^\top\Lambda^{-1}\zeta$. Therefore
\begin{align}
\int_{\mathds S^{d-1}} \!\varphi(\zeta)\,\mathrm d\PnPalm(\zeta)
&=K_\Lambda\int_{\mathds S^{d-1}} \!\varphi(\zeta)
\left(\int_0^{\infty} r^{d}\exp\!\left(-\tfrac12 (\zeta^\top\Lambda^{-1}\zeta)\,r^2\right)\mathrm dr\right)\mathrm d\eta(\zeta),
\notag
\\
&= C_\Lambda\int_{\mathds S^{d-1}} \!\varphi(\zeta)\,\big(\zeta^\top\Lambda^{-1}\zeta\big)^{-\frac{d+1}{2}}\,\mathrm d\eta(\zeta)\,,\notag
\end{align}
where $K_\Lambda:=c_\Lambda(2\pi)^{-d/2}(\det\Lambda)^{-1/2}$ and we used
$\int_0^{\infty} r^{d} e^{-\frac{a}{2} r^2}\,\mathrm dr 
=\tfrac12\,(a/2)^{-\frac{d+1}{2}}\,\Gamma\!\left(\tfrac{d+1}{2}\right)$, absorbed into the constant $C_\Lambda$.
Thus $\PnPalm$ is absolutely continuous with density proportional to
$(\zeta^\top\Lambda^{-1}\zeta)^{-\frac{d+1}2}$. Using the eigen-decomposition $\Lambda={\bm P}^\top{\bm D}{\bm P}$ and setting $\mathbf z:={\bm P}\zeta$, we obtain
\[
\frac{\mathrm d\PnPalm(\zeta)}{\mathrm d\eta(\zeta)}
= C_{\Lambda}\,\big\|{\bm D}^{-\frac12}{\bm P}\zeta\big\|^{-(d+1)}
= C_{\Lambda}\,\Big(\tfrac{z_1^2}{\kappa_1^2}+\cdots+\tfrac{z_d^2}{\kappa_d^2}\Big)^{-\frac{d+1}{2}},
\]
with the normalizing constant $C_{\Lambda}^{-1}=\int_{\mathds S^{d-1}}\big\|{\bm D}^{-\frac12}{\bm P}\zeta\big\|^{-(d+1)}\,\mathrm d\eta(\zeta)$.

\noindent
Now, remark that for $i,j\in[d]$
\begin{align*}
        \EnPalm\big(t_it_j\big)
        &= C_{\kappa}
            \int_{\mathds S^{d-1}}t_it_j \|{\bm D}^{-\frac12}{\bm P}\zeta\big\|^{-(d+1)}\mathrm d\eta(\zeta)    \\
        &= C_{\kappa}
            \int_{\mathds S^{d-1}}\Big(\sum_{\ell=1}^dP_{\ell i}z_\ell\Big)\Big(\sum_{\ell=1}^dP_{\ell j}z_\ell\Big) \|{\bm D}^{-\frac12}\mathbf{z}\big\|^{-(d+1)}\mathrm d\eta(\mathbf{z})   \\
        &= 
        \sum_{\ell=1}^dP_{\ell i}P_{\ell j} \mathcal Z_\ell(\kappa),
\end{align*}
where we used that, by symmetry $z_j\mapsto -z_j$, the expectations of $z_iz_j$ are zero when $i\neq j$.
\end{proof}

\begin{remark}
    \label{rem:estimation_Palm_normalized_eigenvalue}
    Let us comment this result. One can consider the empirical variance-covariance of $N(t)$. In expectation, this statistic is equal to the right hand side of \eqref{eq:covariance_Palm}. Now, consider its spectral decomposition, it gives the model eigen-basis $\bf P$ and its eigenvalues are $\mathcal Z_k(\vv\kappa)$. Hence, one can estimate $\bf P$ and $\mathcal Z_k(\vv\kappa)$ from the eigen-decomposition of $N(t)N(t)^\top$. 
\end{remark}

\begin{remark}
\label{rem:link_Palm_Cabana}
In dimension $d=2$, writing $\Lambda={\bm P}_{-\theta_0}\diag(\kappa_1^2,\kappa_2^2){\bm P}_{\theta_0}$ and $\kappa:=\sqrt{1-\kappa_2^2/\kappa_1^2}\in(0,1)$, we uncover (as in Theorem~\ref{thm:Palm_gradient}) that the density of the angle $\Theta$ of $N(t)$ on $(-\pi,\pi]$ is
\[
f_\Theta(\theta)=C_\kappa\,\big(1-\kappa^2\cos^2(\theta-\theta_0)\big)^{-3/2},
\quad
C_\kappa^{-1}=\int_{-\pi}^{\pi}\big(1-\kappa^2\cos^2\varphi\big)^{-3/2}\,\mathrm d\varphi.
\]
\end{remark}

\subsubsection*{From Palm normalized gradient eigenvalues to model eigenvalues} 

\begin{algorithm}
\caption{High-Dimensional Anisotropy Estimation}
\label{alg:anisotropy_estimation}
\begin{algorithmic}[1]
\Require Level set $\mathcal{L}(\mathcal{T})$ derived from the excursion set $\mathcal{E}(\mathcal{T})$, dimension $d$.
\Ensure Estimated anisotropy parameters $\hat{\vec{\kappa}}$ and directions $\widehat{\bf P}$.

\State \textbf{Step 1: Compute Normals}
\State Extract unit normal vectors $N(t) = \frac{X'(t)}{\|X'(t)\|}$ for all $t \in \mathcal{L}(\mathcal{T})$.

\State \textbf{Step 2: Empirical Covariance}
\State Compute the sample covariance matrix of the normals:
\[
\widehat{\Sigma}_{\rm nPalm} \leftarrow \frac{1}{|\mathcal{L}(\mathcal{T})|} \int_{\mathcal{L}(\mathcal{T})} N(t) N(t)^\top d\mathcal{H}^{d-1}(t)
\]

\State \textbf{Step 3: Spectral Decomposition}
\State Diagonalize $\widehat{\Sigma}_{\rm nPalm}$ to obtain:
\begin{itemize}
    \item Eigenvectors $\widehat{\bf P}$ (estimates the principal directions of anisotropy).
    \item Eigenvalues $\widehat{\mathcal{Z}} = (\hat{\mathcal{Z}}_1, \dots, \hat{\mathcal{Z}}_d)$ (Palm normalized gradient eigenvalues).
\end{itemize}

\State \textbf{Step 4: Convex Optimization (Inversion)}
\State Define the strictly convex functional $\Xi(u)$ as per Eq. \eqref{eq:convex_function_high_dim}:
\[
\Xi(u) = -\frac{2}{d-1} \int_{\mathbb{S}^{d-1}} \left(\sum_{i=1}^{d} u_i z_i^2\right)^{-\frac{d-1}{2}} d\eta(z)
\]
\State Solve \eqref{eq:convex_program_invert_nabla_psi} for $\hat{\pi}$ using Gradient Descent:
\[
\hat{\pi} \leftarrow \arg\min_{u \in \mathbb{R}^d_+} \left\{ \langle \hat{\mathcal{Z}}, u \rangle - \Xi(u) \right\}
\]

\State \textbf{Step 5: Parameter Recovery}
\State Recover anisotropy parameters $\hat{\kappa}_i$ using the closed-form inversion \eqref{eq:inversion_formula_pi}:
\[
\hat{\kappa}_i^2 \leftarrow \frac{1/{\hat{\pi}_i}}{\sum_{j=1}^d (1/{\hat{\pi}_j})}
\]

\State \Return Anisotropy parameters $\hat{\vec{\kappa}} = (\hat{\kappa}_1, \dots, \hat{\kappa}_d)$ and directions $\widehat{\bf P}$.
\end{algorithmic}
\end{algorithm}

Once the $\mathcal{Z}(\vv{\kappa})$ has been estimated from the eigenvalues of some empirical variance-covariance of $N(t)$, one needs to invert the map $\mathcal{Z}\,:\,\Delta_+\to\mathds R^d$ to recover the model eigenvalue $\vv{\kappa}$ (anisotropy parameters). We start with a parametrization of the model eigenvalues~$\vv{\kappa}$ given by the function $\Pi\,:\,\Delta_+\to\mathds R^d$ defined~by 
\begin{align*}
\pi:=\Pi(\vv{\kappa})=(\pi_1,\ldots,\pi_d)\quad\text{where}\quad 
\pi_i := \kappa_i^{-2}\bigg(\int_{\mathds S^{d-1}}\Big(\frac{z_1^2}{\kappa_1^2}+\cdots+\frac{z_d^2}{\kappa_d^2}\Big)^{-{\frac{d+1}{2}}}\mathrm{d}\eta(\mathbf{z})\bigg)^{\frac{2}{d+1}}
\in\mathds R_+\,.
\end{align*}
Observe that given $\pi$ one can recover $\vv{\kappa}\in\Delta_+$ simply by the inversion formula:
\begin{subequations}
\begin{align}
\label{eq:inversion_formula_pi}
    \kappa_i^2 = \frac{\big(\frac{1}{\pi_i}\big)}{\displaystyle\sum_{j=1}^d\frac{1}{\pi_j}}=\frac{1}{\displaystyle\sum_{j=1}^d\frac{\pi_i}{\pi_j}}\,.
\end{align}
Now, consider the concave function
\begin{align}
\label{eq:convex_function_high_dim}
    \Xi\,:\, u\in\mathds R_+^d\mapsto 
    -\frac{2}{d-1}\int_{\mathds S^{d-1}} \,\Big(\sum_{i=1}^d u_i z_i^2\Big)^{-\frac{d-1}2}\mathrm d\eta(\mathbf{z}).
\end{align}
and the convex program 
\begin{equation}
        \label{eq:convex_program_invert_nabla_psi}
            \bar\pi \in \arg\min_{u\in\mathds R_+^d}\big\{\langle\mathcal{Z}(\vv{\kappa}) ,u\rangle -\Xi(u)\big\}\,.
\end{equation}
\end{subequations}
The following theorem shows that $\bar\pi$ is such that $\bar\pi=\Pi(\vv{\kappa})$ and hence, by the inversion formula~\eqref{eq:inversion_formula_pi}, one can recover $\vv{\kappa}$ solving the convex program~\eqref{eq:convex_program_invert_nabla_psi}. A pseudo-algorithm is presented in Algorithm~\ref{alg:anisotropy_estimation}.

\begin{theorem}[Inversion of Palm normalized gradient eigenvalues]
\label{thm:inverse_palm}
Under~\eqref{eq:assumption}, the minimum of~\eqref{eq:convex_program_invert_nabla_psi} exists, is unique, and $\pi=\bar\pi$. Furthermore, the following result on the rate of convergence of gradient descent holds. Let $0<a<b<\infty$ and consider the open convex set $ (a,b)^d$. For any $\pi\in (a,b)^d$, the gradient descent path $u^{(0)},u^{(1)},u^{(2)},\ldots$ of the convex program \eqref{eq:convex_program_invert_nabla_psi}, which starts at point~$u^{(0)}$ with constant gradient step size $h=2/(\alpha+\beta)$, stays in $(a,b)^d$ and 
    \begin{align*}
        \|u^{(k)}-\pi\|&\leq\bigg(\frac{Q-1}{Q+1}\bigg)^k\|u^{(0)}-\pi\|\leq \exp\Big(-\frac{2k}{Q+1}\Big)\|u^{(0)}-\pi\|
    \end{align*}
    where $Q=\beta/\alpha={d(d+1)} {(\frac{b}{a})^{\frac{d+3}2}}\int_{\mathds S^{d-1}}z_1^4 \mathrm d\eta(\mathbf{z})$.
\end{theorem}

\begin{proof}
First, note that the Fenchel-Legendre dual of $\Theta$ and its gradient\footnote{We denote the gradients of a deterministic function $f$ by $\nabla f$ and those of random fields $X(\cdot)$ by $X'(\cdot)$.} are given by 
\begin{align*}
    \Xi^\star(v)&=
    \inf_{u\in\mathds R_+^d}
    \big\{\langle v,u\rangle -\Xi(u)\big\}\in\mathds R\cup\{-\infty
    \}\,,\\
     \nabla\Xi(u)
     &= \Omega(u)\,,
\end{align*}
where
\begin{align*}
    \Omega&\,:\, \pi\in\mathds{R}^d\mapsto
    \bigg(            
        {\int_{\mathds S^{d-1}}z_\ell^2 \,\Big(\sum_{i=1}^d\pi_i z_i^2\Big)^{-\frac{d+1}2}\mathrm d\eta(\mathbf{z})}
    \bigg)_\ell\in\mathds{R}^d\\
    \text{is such that }\Omega&(\Pi(\vv{\kappa}))=\mathcal{Z}(\vv{\kappa})\,,\quad \forall\vv{\kappa}>0\,,
\end{align*}
where $\vv{\kappa}>0$ means that all the coordinates $\kappa_i$ of $\vv{\kappa}$ are positive. The next lemma shows that 
\[
\nabla\Xi^\star(\mathcal{Z}(\vv{\kappa}))=\Pi(\vv{\kappa})\,,
\]
the minimum of~\eqref{eq:convex_program_invert_nabla_psi} exists, is unique, and $\pi=\bar\pi$.
\begin{lemma} \label{lem:Psi_strongly_concave}
The function $\Xi$ is strictly concave on a open convex set, see \cite[Proposition 3.1.1, Page 140]{nesterov2018lectures} for a definition. Hence, for all $v\in\mathds R^d$ such that there exists $u_v\in\mathds{R}_+^d$ with $v=\nabla\Xi(u_v)$, it holds that $\nabla\Xi^\star(v)$ is the unique solution to the following convex program 
\begin{align*}
    \nabla\Xi^\star(v)=\arg\min_{u\in\mathds R_+^d}\big\{\langle v,u\rangle -\Xi(u)\big\}\,,
\end{align*}
and $\nabla\Xi^\star(v)=u_v$. The points $u_v$ and $v$ are conjugates by the diffeomorphisms $\nabla\Xi$ and $\nabla\Xi^\star$. 
\end{lemma}

\begin{proof}
For any unit norm vector $w\in\mathds S^{d-1}$, the Hessian of $\Xi$ at point $u\in\mathds R_+^d$ satisfies
\begin{align*}
        w^\top\nabla^2\Xi(u)w & = -\frac{d+1}{2} \sum_{i,j}\int_{\mathds S^{d-1}}w_iz_i^2w_jz_j^2 \,(\sum_{k=1}^du_k z_k^2)^{-\frac{d+3}2}\mathrm d\eta(\mathbf{z})\\
        &= -\frac{d+1}{2} \int_{\mathds S^{d-1}}(\sum_{i=1}^dw_iz_i^2)^2 \,(\sum_{k=1}^du_k z_k^2)^{-\frac{d+3}2}\mathrm d\eta(\mathbf{z})<0
\end{align*}
Hence $\Xi$ is a strictly concave function. The uniqueness and existence follows from standard results in convex optimization \citep[Theorem 2.1.1]{nesterov2018lectures}. 
\end{proof}

Now, we dig into bounding the Hessian of $\Theta$ with the following lemma.
\begin{lemma}
\label{lem:smooth_param_express}
Let $0<a<b<\infty$. Restricted to the open convex set $(a,b)^d$, the function $\Xi$ is $\beta$-smooth and $\alpha$-strongly convex \citep[Proposition 2.1.11, Page 75]{nesterov2018lectures} with $\alpha={b^{-\frac{d+3}2}}/{2}$ and $\beta = ({d(d+1)}/{2}) {a^{-\frac{d+3}2}}\int_{\mathds S^{d-1}}z_1^4 \mathrm d\eta(\mathbf{z})$.
\end{lemma}

\begin{proof}
    For any unit norm vector $w\in\mathds S^{d-1}$, the Hessian of $\Xi$ at point $u\in\mathds R_+^d$ satisfies
    \begin{align*}
        w^\top\nabla^2\Xi(u)w 
        &= -\frac{d+1}{2} \int_{\mathds S^{d-1}}(\sum_{i=1}^dw_iz_i^2)^2 \,(\sum_{k=1}^du_k z_k^2)^{-\frac{d+3}2}\mathrm d\eta(\mathbf{z})\\
        &\leq -\frac{d+1}{2}b^{-\frac{d+3}2}\int_{\mathds S^{d-1}}(\sum_{i=1}^dw_iz_i^2)^2\mathrm d\eta(\mathbf{z})\\
        &\leq -\frac{d+1}{2}b^{-\frac{d+3}2}\min_{\omega\in\mathds R^d}\int_{\mathds S^{d-1}}(\sum_{i=1}^d\omega_iz_i^2)^2\mathrm d\eta(\mathbf{z})
        \\
        &=-\frac{d+1}{2d}b^{-\frac{d+3}2}\\
        &\leq -\frac{b^{-\frac{d+3}2}}{2}
    \end{align*}
    Hence $\Xi$ is a strongly concave function. Now, for any unit norm vector $w\in\mathds S^{d-1}$, the Hessian of $\Xi$ at point $u\in\mathds R_+^d$ satisfies
    \begin{align*}
        w^\top\nabla^2\Xi(u)w 
        &= -\frac{d+1}{2} \int_{\mathds S^{d-1}}(\sum_{i=1}^dw_iz_i^2)^2 \,(\sum_{k=1}^du_k z_k^2)^{-\frac{d+3}2}\mathrm d\eta(\mathbf{z})\\
        &\geq -\frac{d+1}{2} {a^{-\frac{d+3}2}}\int_{\mathds S^{d-1}}(\sum_{i=1}^dz_i^4) \mathrm d\eta(\mathbf{z})\\
        &=-\frac{d(d+1)}{2} {a^{-\frac{d+3}2}}\int_{\mathds S^{d-1}}z_1^4 \mathrm d\eta(\mathbf{z})
    \end{align*}
    Hence $\Xi$ is a smooth function.
\end{proof}

\noindent
To finish the proof, apply \cite[Theorem 2.1.15, Page 81]{nesterov2018lectures} and the only technical point is to prove that the gradient descent stays in $(a,b)^d$. This is clear since at each step the distance to the minimum decreases and the set $(a,b)^d$ is convex.
\end{proof}

\paragraph{Consequences and practical reading}
The inversion theorem has two immediate consequences.

\emph{Identifiability and algorithmics:} for any observed Palm normalized-gradient eigenvalues $\mathcal Z(\vv{\kappa})$ there is a unique $\pi=\Pi(\vv{\kappa})$ solving the convex program, hence a unique $\vv{\kappa}$ by the closed-form inversion \eqref{eq:inversion_formula_pi}. Moreover, plain gradient descent with a constant step enjoys global $Q$-linear convergence with an explicit rate, so the map $\mathcal Z\mapsto\vv{\kappa}$ can be computed robustly.

\emph{Stability:} the rate depends only on the smoothness/strong-convexity moduli on the box $(a,b)^d$ through $Q=\beta/\alpha$, hence the conditioning is explicit; larger boxes (or larger anisotropy) increase $Q$, but the convergence remains linear.

\emph{Ill-conditioned problems:} Regarding the hypothesis $\pi\in(a,b)^d$ in terms of the model eigenvalues $(\kappa_i)$, recall that
\[
\pi_i \;=\; \frac{s(\vec{\kappa})}{\kappa_i^2},\qquad
s(\vec{\kappa}):=\Bigg(\int_{\mathbb S^{d-1}}\Big(\sum_{j=1}^d\frac{z_j^2}{\kappa_j^2}\Big)^{-\frac{d+1}{2}}\,\mathrm d\eta(\mathbf z)\Bigg)^{\!\frac{2}{d+1}}.
\]
Let $\kappa_{\min}:=\min_i\kappa_i>0$ and $\kappa_{\max}:=\max_i\kappa_i$. Since
$\sum_j z_j^2/\kappa_j^2\in[1/\kappa_{\max}^2,\,1/\kappa_{\min}^2]$ on $\mathbb S^{d-1}$, one has
$s(\vec{\kappa})\in[\kappa_{\min}^2,\kappa_{\max}^2]$ and therefore, for all $i \in \{1, \dots, d\}$,
\[
\frac{1}{r^2}=\frac{\kappa_{\min}^2}{\kappa_{\max}^2}\;\le\;\pi_i\;\le\;\frac{\kappa_{\max}^2}{\kappa_{\min}^2}=r^2\quad\text{where}\quad r := \frac{\kappa_{\max}}{\kappa_{\min}} \ge 1\,.
\]
The previous inequality implies that the true parameter $\pi$ lies within the interval $[r^{-2}, r^2]$. Consequently, the optimization box $(a,b)^d$ must satisfy $a < r^{-2}$ and $b > r^2$.

While the ratio $r$ is fixed for any specific realization of the field, it influences the convergence rate of the estimator. A convenient choice of bounds that guarantees the inclusion of $\pi$ is to set $a$ and $b$ such that the ratio $b/a$ is of the order of $r^4$ (for instance $a = \frac{1}{2}r^{-2}$ and $b = 2r^2$). With such a choice, the condition number $Q = \beta/\alpha$ derived in Lemma \ref{lem:smooth_param_express} satisfies:
\[
Q \;=\; C_d \left(\frac{b}{a}\right)^{\frac{d+3}{2}} \;=\; \mathcal{O}_{r\shortrightarrow\infty}\big(r^{2(d+3)}\big),
\]
where $C_d$ is a positive constant depending only on the dimension $d$. This yields the following convergence bound:
\[
\|u^{(k)}-\pi\|\le\Big(\frac{Q-1}{Q+1}\Big)^k\|u^{(0)}-\pi\|
\quad\text{with}\quad
\frac{Q-1}{Q+1}=1-\frac{2}{C_d}r^{-2(d+3)}+o_{r\shortrightarrow\infty}\big(r^{-2(d+3)}\big)
\,.
\]
This result demonstrates that higher anisotropy (a larger $r$) increases the condition number $Q$ polynomially, which slows down the linear convergence rate, but does not affect the uniqueness of the solution nor the global convergence of the algorithm. In practice, if prior bounds on $(\kappa_i)$ are unavailable, one may select a sufficiently wide box (e.g., $a\ll 1 \ll b$) and project the gradient steps onto $[a,b]^d$; the theorem holds as long as the true $\pi$ is contained within $(a,b)^d$.

\subsection{Consistency and asymptotic distributions of estimators} \label{s:coco}
Recall that $\cT_n=(-n,n)^d, \ n\in\N$, and define the contour integral
\begin{equation*} \label{eq:non_asymptotic_Palm}
    I_n(f) :=(2n)^{-d} \int_{\lev(\cT_n)} f\left(\frac{X'(t)}{|X'(t)|}\right) \ \d\cH^1(t),
\end{equation*}
where $f(\cdot) $ is a bounded, continuous real-valued function defined on the sphere $\mathds S^{d-1}$. 
Note that choosing $f(\zeta_1,\zeta_2)=\zeta_1$ (resp. $\zeta_2$) we get $\cC(\cT_n)$ (resp. $\cS(\cT_n)$) in dimension $2$, while choosing $f(\zeta_1,\dots,,\zeta_d)=\zeta_i\zeta_j$ we get $\cC_{ij}(\cT_n)$ in dimension $d\geq2$, see \eqref{e:cij}. Choosing $f=1$ yields $|\lev(\cT_n)|$ in any dimension. In \citep[Th. 4.1]{berzin} it is proved that $I_n(f)$ converges almost surely towards $\E[I_1(f)]$\footnote{In \citep{berzin} Theorems 4.1 and 4.7 are written for $\R^2$, but it is straightforward to see that they hold in $\R^d$. The restriction to dimension $2$ in \citep{berzin} comes from the solution to the implicit equation defining the estimators, see her Remark on page 8.}. As a corollary we get:
\begin{theorem}
With the notation of Section \ref{s:framework} and Assumption~\eqref{eq:assumption}, if the covariance $r(\cdot)$ of $X(\cdot)$ vanishes at infinity $($i.e., $r(t)\to0$ as $\|t\|\to\infty)$, then 
\begin{enumerate}
 \item When $d=2$ and $\kappa\neq0$, $\widehat{\theta}_0$ and $\widehat{\kappa}_C$ defined in~\eqref{e:theta:hat} and \eqref{e:k:hat} are strongly consistent estimators of~$\theta_0$ and $\kappa$ respectively. 
 \item When $d\geq 2$, the matrix defined by the entries $\Big(\frac{\cC_{ij}(\cT_n)}{|\lev(\cT_n)|}\Big)_{i,j=1,\dots,d}$, see \eqref{e:cij}, is a strongly consistent estimator of the matrix ${\bm P}^\top \diag(\mathcal Z(\vv\kappa)){\bm P}$.
\end{enumerate}
\end{theorem}
\noindent 
Furthermore, \cite{berzin} proves a Central Limit Theorem for the functionals $I_n$.
\begin{theorem} {\cite[Th. 4.7]{berzin}} \label{t:coco}
With the notation of Section \ref{s:framework} and Assumption~\eqref{eq:assumption}, assume that there exists an integrable function $\Psi:\R^d\to\R$ vanishing at infinity which dominates each derivative of the covariance $r(\cdot)$ of $X(\cdot)$ up to the second order, namely, 
\[ 
  |r(t)|, |r_i(t)|, |r_{ij}(t)|\le \Psi(t),\ t\in\R^d,\  i,j=1,\dots,d.
\]
Then, for any bounded, continuous real-valued function $f$ defined on $\mathds S^{d-1}$, there exists $v(f)\in[0,\infty)$ such that 
\[
  J_n(f):= (2n)^{d/2} (I_n(f) - \E(I_n(f))) \Rightarrow  \mathcal N(0, v^2(f)),
\]
where $ \Rightarrow $ denotes the convergence in distribution. 
\end{theorem}
It is worth mentioning that, assuming in addition that $\int_{\R^d} r(t) >0$ and that $f$ is of constant sign, Berzin proves that the asymptotic variance $v$ of $J_n$ in dimension $2$ is strictly positive.

Some remarks and examples are in order. Though the assumptions on the covariance seem quite awful, on the one hand they appear naturally in the proofs and on the other hand they are satisfied by most of the usual random fields.
\begin{remark}{\bf Comments on the hypotheses and the proof of Theorem \ref{t:coco}.}
In the first place we would like to mention that in the literature there exist variations of the so-called Arcones type condition, that is, on the integrability of the function $\Psi$, replacing $L^1$ by $L^2$ or even by $L^4$. Secondly, the domination of the covariance and its derivatives by such a function $\Psi$ implies that they are in $L^p(\R^d)$ for all $p\in\N$. By the inversion formula, $X(\cdot)$ admits a continuous spectral density which is positive at $0$. Finally, from the methodological point of view, Berzin's proof is based on the Peccati-Tudor method \citep{peccati-tudor}, see also \citep{nourdin-peccati}, which consists on expanding $J_n$ in the so-called chaotic expansion and applying a simplified version of the classical Moments Method. The Arcones inequality \citep[L.1.]{arcones} plays a key role. This method yields that, in order to state the finiteness of the asymptotic variance and the asymptotic normality of $J_n$, it suffices to control integrals of products of derivatives of the covariance $r$. The domination by $\Psi$ allows one to control these integrals. 
\end{remark}

\begin{remark}{\bf Examples of covariances verifying the hypotheses of Theorem \ref{t:coco}.}
The hypotheses of the theorem are quite weak in the sense that they are satisfied by most of the usual Gaussian random fields. The following list is not exhaustive. \cite{berzin} points out the \emph{Powered exponential covariance} $r(t)=C\exp\{-\alpha\|t\|^2\}$, $C,\alpha>0$, including the celebrated \emph{Bargmann-Focks} model; the \emph{Generalized Cauchy covariance} $r(t)= C(\alpha^2+\|t\|^2)^{-\nu}$, $C,\alpha,\nu>0$ and the \emph{Whittle-Matérn covariance} $r(t)=C(\alpha\|t\|)^\nu K_\nu(\alpha\|t\|)$, $C,\alpha>0$ and $\nu>2$. 
Besides, in the context of random waves on~$\R^3$, \cite{del} presents some Gaussian random fields including the \emph{Gamma type} covariance $r(x)=\frac 1p(1+\frac{|x|^2}{\beta^2})^{-p}\,\sum_{1\le j\le p;\,j\,odd}(-1)^{(j-1)/2}\,\binom{p}{j}\,\beta^{-(j-1)}\,|x|^{j-1}$, $p\in\N$, $\beta>0$ and $p\geq d$, which is a variation of the Generalized Cauchy covariance and satisfy these hypotheses; and the \emph{Black Body Radiation} and the \emph{Monochromatic Random Waves} models which covariances are not integrable, but they still provide us further examples by taking convenient powers of their covariances. 
\end{remark}

Theorem \ref{t:coco} is sufficient to obtain, by polarization, the joint normality of the integrals involved in the Contour-Method and in its generalization to higher dimensions. This gives by the Delta-method \cite[Ch.3]{vandervaart} the normality of the estimator $\widehat\kappa_{\rm{C}}$ given by \eqref{e:k:hat}.
\begin{corollary}
\label{cor:assymptotic_normality}
Under the conditions of Theorem \ref{t:coco}, we have in dimension $2$ that
\begin{enumerate}
\item $\cC(\cT_n)$, $\cS(\cT_n)$ and $\lev(\cT_n)$ are asymptotically jointly Gaussian.
\item $\cC(\cT_n)$ and $\cS(\cT_n)$ are asymptotically independent.
\item If $\kappa\neq0$, the estimator $\widehat\kappa_{\rm{C}}$  given by \eqref{e:k:hat} is asymptotically normal.
\item $Q_{n,N}$ defined in \eqref{e:qn} is asymptotically $\chi^2(2)$. 
\end{enumerate}
Furthermore, in higher dimension $d$ we have
\begin{enumerate}[resume]
 \item The $\cC_{ij}(\cT_n):1\leq i,j\leq d$ defined in \eqref{e:cij} are asymptotically jointly Gaussian. 
\end{enumerate}
\end{corollary}

Let us finish this section reviewing the LKC case. The (separate) asymptotic normality of the Lipschitz-Killing curvatures was proved in \citep{hung} in the case of the volume, in \citep{estrade-leon} for the Euler-Poincaré Characteristic and in \citep{muller2017} and \citep{kratz-Vadlamani} for the remaining curvatures. However, the limit joint normality of all the LKC seems more challenging, as noted in \citep{kratz-Vadlamani}, and it has not been addressed in the literature so far. In particular, the integral formulas representing the volume on one side and the remaining LKC on the other are of different nature. It is our intention to generalize, in a future paper, Theorem 4.7 in \citep{berzin} to an integral of the form 
\[
  J_n(f) :=(2n)^{-d} \int_{\lev(\cT_n)}  f(X'(t), X''(t)) \ \d\cH^1(t).
\]
This would permit us to obtain the joint normality of the quantities considered here, possibly with the exception of the volume.

\subsection{Testing procedure of anisotropy}  \label{s:chi2}

To apply the results of Berzin, we assume the conditions of Theorem \ref{t:coco}. All asymptotic below are under~$\mathds H_0$ as $n\to\infty$. We restrict to $d=2$ for clarity; higher dimensions are similar.

\begin{theorem}[Asymptotic chi-square limit under $\mathds H_0$]\label{thm:chi2}
Under the assumptions of Theorem~\ref{t:coco} and  for $d=2$, it holds that
\[
\frac{\cC(\cT_n)^2 + \cS(\cT_n)^2}{|\cT_n|} \Rightarrow V^2\,\chi^2(2),
\]
where $V^2$ is a model-dependent variance given in \citep{berzin}.
\end{theorem}

\begin{remark}
The scale factor $V^2$ is a nuisance parameter: it depends on the underlying model but not on the presence of anisotropy under $\mathds H_0$. Estimating $V^2$ is therefore required to obtain a pivotal statistic for testing.
\end{remark}

We estimate $V^2$ by spatial sub-sampling. Consider a quasi-partition of $\cT_n$ into $N^2$ isometric rectangles of volume $|\cT_n|/N^2$, with local statistics $\cC_{i}$ and $\cS_{i}$. A natural (approximately unbiased) estimator is
\begin{subequations}
\begin{equation} \label{e:V}
\hat V^2 = \frac{1}{2(N^2-1)}\sum_{i=1}^{N^2}[(\cC_{i}-\bar \cC)^2+(\cS_{i}-\bar \cS)^2]
\simeq \frac{V^2 |\cT|}{N^2},
\end{equation}
where $\bar \cS$ and $\bar \cC$ denote the means of $(\cS_i)_i$ and $(\cC_i)_i$, respectively and the equivalence is obtained as $n,N \to \infty$  with $N=o(n)$.

\begin{remark}
In practice, choose $N$ so that $N^2$ is large enough (many blocks) while each block still contains many sampling points; this balances the variance and the bias of $\hat V^2$.
\end{remark}

\noindent
As a consequence we have  if we assume that $n,N\to\infty$ with $N=o(n)$
\begin{equation} \label{e:qn}
 Q_{n,N} :=\frac{\cS(\cT_n)^2 + \cC(\cT_n)^2}{N^2 \hat V^2}  \Rightarrow \chi^2(2),
\end{equation}
\end{subequations}
\begin{theorem}[Asymptotic p-value under $\mathds H_0$]\label{thm:pvalue}
Assume~\eqref{eq:assumption}. Let $F_{\chi^2(2)}$ denote the cdf of a $\chi^2$ distribution with 2 degrees of freedom and define the test statistic
\[
\alpha := F_{\chi^2(2)}\big(Q_{n,N}\big).
\]
Under $\mathds H_0$, one has $\alpha \Rightarrow \mathrm{Unif}(0,1)$ as $n,N\to\infty$ with $N=o(n)$. Equivalently, the upper-tail p-value $1-\alpha$ is also asymptotically $\mathrm{Unif}(0,1)$.
\end{theorem}

\section{Numerical experiments} 
\label{s:num}

 In this section  we present for the first time a numerical evaluation of the contour method and a comparison with LKC method and the so-called ‘‘full observation'' where one has access to the random field and compute the covariance of the gradient. This  will serve latter as an oracle to compare our methods. All our experiments are publicly available on a Github repository: \url{https://github.com/ydecastro/COMETE-Contour-method-Gaussian-Random-Fields/}.

\subsection{Numerical framework}

We consider stationary Gaussian random fields on the square $\cT=[0,200]^2$ with covariance
\[\Cov(X(0),X(t)) = \exp\Big(-\tfrac12\big(a^2 t_1^2 + a^{-2} t_2^2\big)\Big), \quad a>1.\]
By direct differentiation, $\kappa_1^2=a^2$, $\kappa_2^2=a^{-2}$ and thus $\kappa^2 = 1-a^{-4}$. We explore three anisotropy levels corresponding to $\kappa\in\{0,0.5,0.9\}$. Without loss of generality, we set $\mu=0$, $\sigma=1$, and consider levels $u\in\{0,1,2\}$. We draw $2,000$ realizations under $\mathds H_0$ and $5,000$ under the alternative. Each realization is observed on a $1,000\times1,000$ grid; the level set is extracted with the Python \texttt{Contour} routine, and $10^7$ points are equispaced on the level set (with respect to the curvilinear abscissa) to compute integrals such as $\cS$ (and $\cC$) and the Gaussian curvature.

\paragraph{Computing level sets and normals.}
We extract the level set $\mathcal{L}(\mathcal{T}) = \{t \in \mathcal{T} : X(t) = u\}$ from the discretized random field using the \texttt{contour} routine from the \texttt{matplotlib} library, which applies the marching squares algorithm to generate polygonal approximations of the iso-contours. This procedure yields a collection of paths (sequences of vertices). To ensure high-precision numerical computation of the integrals $\mathcal{C}$ and $\mathcal{S}$, we resample these polygonal paths uniformly with respect to the arc-length, typically targeting a total of $10^7$ points across the observation window. At each point $t$ on the resampled curve, the unit tangent vector is computed via centered finite differences of the neighboring vertices, and the unit normal vector $N(t)$ is obtained by a rotation of $\pi/2$. We note that while the marching squares algorithm generally preserves the orientation of closed loops (gradient direction), the statistics $\mathcal{C}$ and $\mathcal{S}$ depend on $2\Theta(t)$ and are therefore invariant to a global sign flip of the normal $N(t) \to -N(t)$, making the method robust to orientation conventions.

\paragraph{Computing the Lipschitz–Killing curvatures.} The accuracy of the LKC method relies heavily on how the Euler characteristic (the third Lipschitz–Killing curvature in 2D) is estimated. While the standard topological approach—counting connected components minus holes—is intuitive, it proves numerically unstable on discrete grids due to sensitivity to local noise. Instead, we adopt a geometric approach based on the Gauss–Bonnet theorem. By computing the integral of the Gaussian curvature along the boundary as in \eqref{e:gauss}, we obtain a significantly more robust estimator.

  \begin{figure}[!ht]
     \centering
     \includegraphics[width=0.38\linewidth]{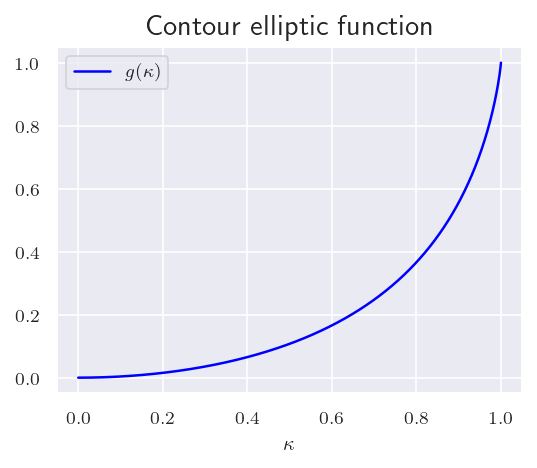}\hspace*{1cm}
     \includegraphics[width=0.38\linewidth]{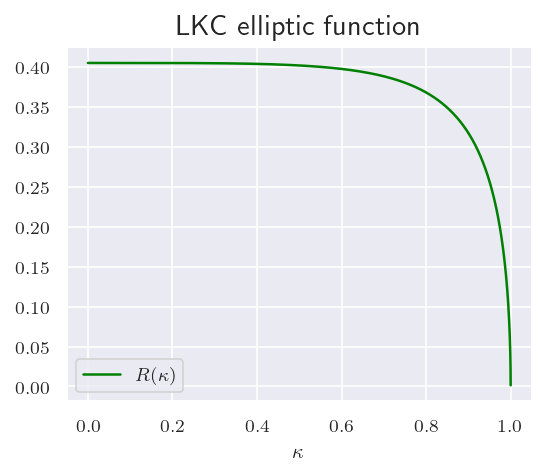}
      \caption{Elliptic functions linking anisotropy and statistics: $g(\cdot)$ for the Contour method (left) and $R(\cdot)$ for the LKC method (right).}
     \label{fig:inverse_MB}
 \end{figure}

\paragraph{Inverting elliptic functions.}
The theoretical functions linking the anisotropy parameter $\kappa$ to the observed statistics—$g(\cdot)$ for the Contour method and $R(\cdot)$ for the LKC method—are defined by complete elliptic integrals. We evaluate these integrals numerically. As illustrated in Figure~\ref{fig:inverse_MB}, both functions are strictly monotone over the domain $\kappa \in [0, 1)$, which ensures that the estimators are well-defined and unique. In practice, we recover the anisotropy estimate $\hat{\kappa}$ by inverting these mappings using a standard root-finding algorithm, such as the bisection method (dichotomy) or Brent's method.

\subsection{Estimation of the anisotropy ratio and principal direction}
 We investigate the distribution of the estimates $\widehat\kappa$ of the anisotropy ratio $\kappa$ in dimension $d=2$. In the simulations (see the accompanying notebooks\footnote{\url{https://github.com/ydecastro/COMETE-Contour-method-Gaussian-Random-Fields/}}), we recover~$\kappa$ using four complementary approaches and compare their behavior:
 
 \begin{itemize}
 \item {\bf Contour method.} We extract unit normals along the estimated level set and compute the trigonometric aggregates $\cC$ and $\cS$ built from $\cos(2\Theta)$ and $\sin(2\Theta)$. Their normalized magnitude is a monotone function $g(\kappa)$ of the anisotropy ratio; hence we obtain the estimator $\widehat\kappa_{\rm C}$ by numerically inverting $g(\cdot)$ (cf. Fig.~\ref{fig:inverse_MB}). The principal direction is estimated by $\widehat\theta_0=\tfrac12\arctan(\cS/\cC)$ using the {\tt atan2} function of Python which is more robust.
 
 \item {\bf LKC method.} From the excursion geometry we compute a robust Lipschitz--Killing curvature-based statistic (using Gaussian curvature) which depends monotonically on $\kappa$ through an elliptic function $R(\kappa)$. The estimator $\widehat\kappa_{\rm LKC}$ is defined by inverting the elliptic function $R(\cdot)$ (cf. Fig.~\ref{fig:inverse_MB}).
 
 \item {\bf $\chi^2$ calibration (isotropy test).} Using a partition of the domain into $N^2$ cells, we form $Q_{n,N}$ in~\eqref{e:qn} and the $p$-value $\alpha=F_{\chi^2(2)}(Q_{n,N})$ (Theorem~\ref{thm:pvalue}). This yields a well-calibrated test for quasi-isotropy which do not require any prior knowledge of the covariance function of the random field (model-free method). We report it alongside point estimators (Contour and LKC model-based methods) to assess significance; it is not used to define $\widehat\kappa$ directly.
 
 \item {\bf Gradient covariance (oracle benchmark).} When the full field is available on a grid, we compute finite-difference gradients, estimate $\widehat\Lambda=\widehat{\Var}(X')$, and take its eigenvalues $\hat\lambda_1\ge\hat\lambda_2$. In $d=2$, the anisotropy parameter is recovered by
 $\displaystyle\widehat\kappa_{\rm grad}=\sqrt{1-\hat\lambda_2/\hat\lambda_1}$,
 and the principal direction is the leading eigenvector. This serves as a reference in simulations.
 \end{itemize}

\subsection{Contour cosine and sine align along the true angle}
\label{sec:Cabana}

We now illustrate the core mechanism of the Contour method by visualizing the empirical statistics used to estimate anisotropy. This experiment serves to empirically validate the theoretical expectations derived from the Palm distribution (Theorem~\ref{thm:Palm_gradient} and Remark~\ref{rem:link_Palm_Cabana}) and demonstrates how the ``signal'' of anisotropy emerges from the geometry of the level sets.

Using the simulation framework described before, we generate realizations of the Gaussian field on~$\mathcal{T}$ with a fixed principal direction $\theta_0=1$ radian. For each realization, we compute the \emph{normalized} trigonometric statistics $(\hat\cC, \hat\cS)$ by averaging the unit normals along the level set:
\[
    \hat\cC=\int_{\mathcal{L}(\mathcal{T})} \cos(2\Theta(t))\,\frac{\d\mathcal{H}^1(t)}{|\mathcal{L}(\mathcal{T})|},\qquad
    \hat\cS=\int_{\mathcal{L}(\mathcal{T})} \sin(2\Theta(t))\,\frac{\d\mathcal{H}^1(t)}{|\mathcal{L}(\mathcal{T})|}\,.
\]
These quantities represent the empirical mean of the vector $(\cos 2\Theta, \sin 2\Theta)$ with respect to the uniform measure on the level curve.

 \begin{figure}[!ht]
     \centering
     \includegraphics[width=\linewidth]{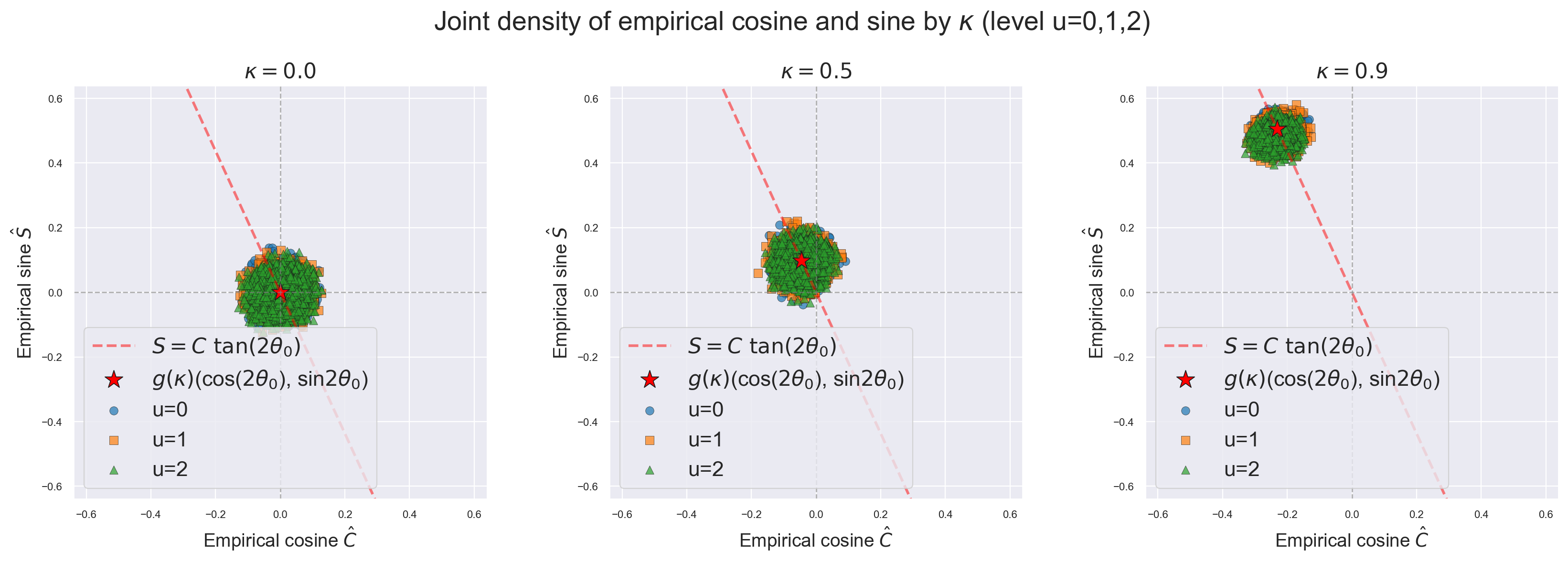}
        \caption{Scatter plots of $(\hat\cC,\hat\cS)$ over many realizations for $\kappa\in\{0,0.5,0.9\}$. Under isotropy ($\kappa=0$), the cloud is ‘‘distributed'' around $(0,0)$. Under anisotropy, the cloud is ‘‘distributed'' around $g(\kappa)(\cos2\theta_0,\sin2\theta_0)$, enabling estimation of $\theta_0$ via $\tfrac12\operatorname{arctan}(\hat\cS,\hat\cC)$ and of $\kappa$ via $g^{-1}$.}
     \label{fig:Cabana_MB}
 \end{figure}

\paragraph{Recovering parameters from the cloud of points}
Figure~\ref{fig:Cabana_MB} displays the distribution of the vector $(\hat\cC,\hat\cS)$ across many realizations. The scatter plots reveal the link between the observed geometry and the model parameters derived in Section~\ref{s:cont}:
\begin{itemize}
    \item \textbf{Direction ($\theta_0$):} The cloud of points is oriented along the angle $2\theta_0$. Consequently, $\theta_0$ can be recovered using the four-quadrant inverse tangent function, defined as $\widehat\theta_0 = \tfrac{1}{2}\operatorname{atan2}(\hat\cS, \hat\cC)$.\footnote{The function $\operatorname{atan2}(y, x)$ computes the principal value of the argument function applied to the complex number $x+\mathrm{i}y$, returning an angle in $(-\pi, \pi]$. Unlike the standard arctan, it distinguishes between all four quadrants.}
    \item \textbf{Intensity ($\kappa$):} The distance of the cloud from the origin corresponds to the anisotropy strength. As predicted by \eqref{e:cabana}, the theoretical mean is given by $\mathbb{E}[(\hat\cC, \hat\cS)] \approx g(\kappa)(\cos 2\theta_0, \sin 2\theta_0)$. Since $g(\kappa)$ is strictly increasing (Lemma~\ref{l:g}), the magnitude $\hat{\mathcal{F}} = \sqrt{\hat\cC^2+\hat\cS^2}$ serves as a consistent estimator for $\kappa$ via inversion.
\end{itemize}

\paragraph{Validation of the asymptotic behavior}
The visualization confirms the robustness of the Contour method. Under quasi-isotropy ($\kappa=0$), the cloud concentrates around the origin $(0,0)$, validating the null hypothesis behavior where no direction is preferred. Conversely, as $\kappa$ increases (from $0.5$ to $0.9$), the cloud shifts clearly away from the origin. This shift empirically validates the Palm density derived in Theorem~\ref{thm:Palm_gradient}, showing that the normals concentrate preferentially orthogonal to the direction of anisotropy. The stability of the cloud's angle across realizations supports the use of the Contour method for precise direction estimation, even at moderate observation windows.

 \begin{figure}[!ht]
     \centering
     \includegraphics[width=\linewidth]{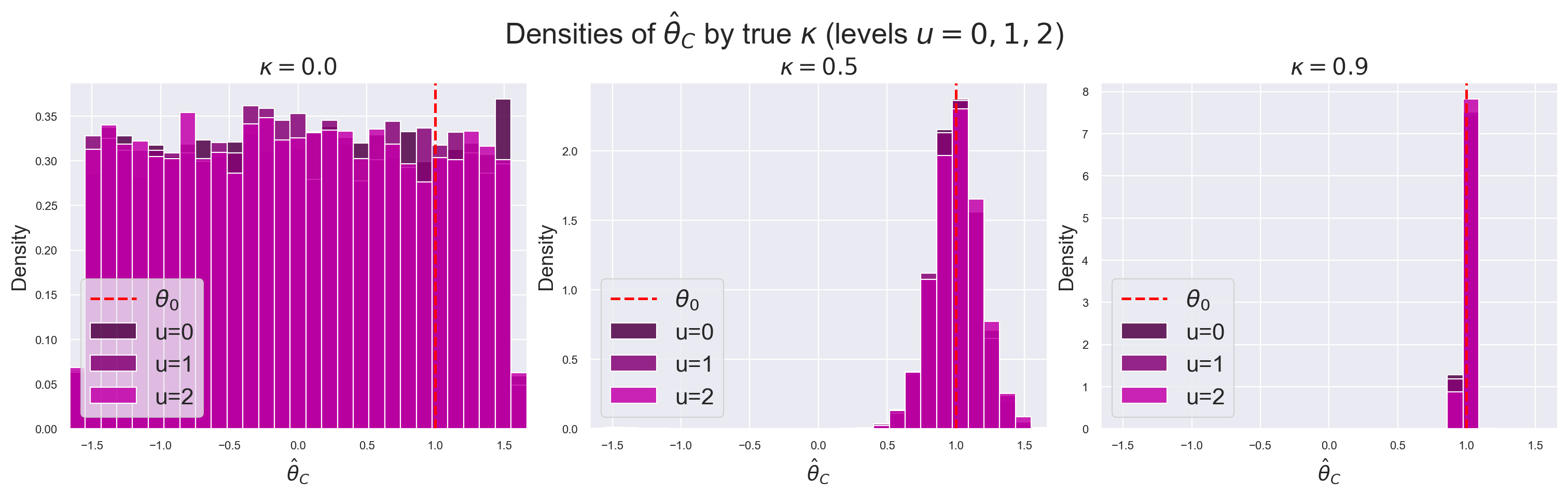}
        \caption{Densities of the angle estimation $\widehat\theta_{\rm C}$ for different values of $\kappa$ and levels $u$.}
     \label{fig:angle_cabana}
 \end{figure}

\paragraph{Estimation of the angle of isotropy}
In our study, the true angle is $\theta_0=1$ radian. We estimate this angle using $\widehat\theta_{\rm C} = \tfrac{1}{2}\operatorname{atan2}(\hat\cS, \hat\cC)$. In Figure~\ref{fig:angle_cabana} we uncover that the angle is uniformly distributed under the null hypothesis (isotropy), for all levels. As the anisotropy increases, the Contour estimator $\widehat\theta_{\rm C}$ remains unbiased and precise, with a variance that is \emph{remarkably stable across different observation levels $u$}

\subsection{Contour is more accurate than LKC}
\label{sec:compare_LKC}

We compare the point-estimation accuracy of the anisotropy parameter using the Contour estimator $\widehat\kappa_{\rm C}$ (obtained by inverting $g(\kappa)$ from the trigonometric summaries of normals) and the LKC estimator $\widehat\kappa_{\rm LKC}$ (obtained by inverting the theoretical function $R(\kappa)$ from Lipschitz–Killing curvatures) in Figure~\ref{fig:marginal_density_kappa_estimates_MB}. For each configuration, we run many realizations and report the joint empirical distributions of $(\widehat\kappa_{\rm C},\widehat\kappa_{\rm LKC})$.

 \begin{figure}[!ht]
     \centering
     \includegraphics[width=\linewidth]{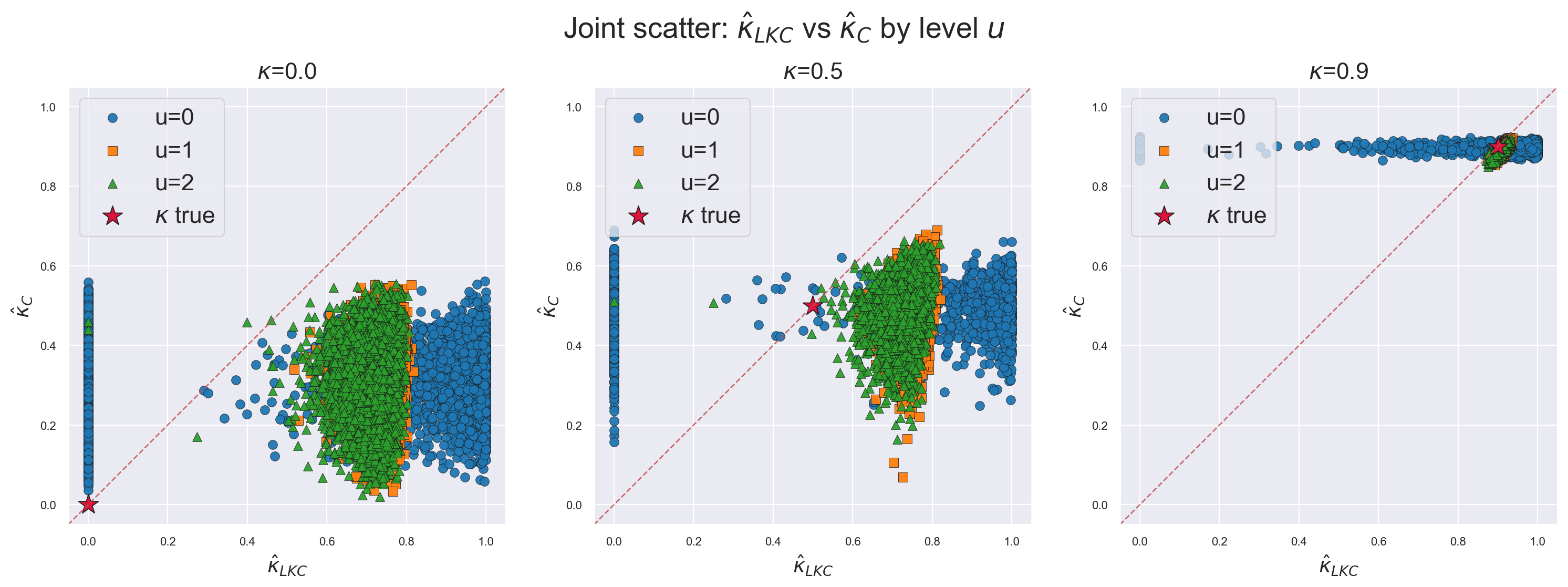}
     \includegraphics[width=\linewidth]{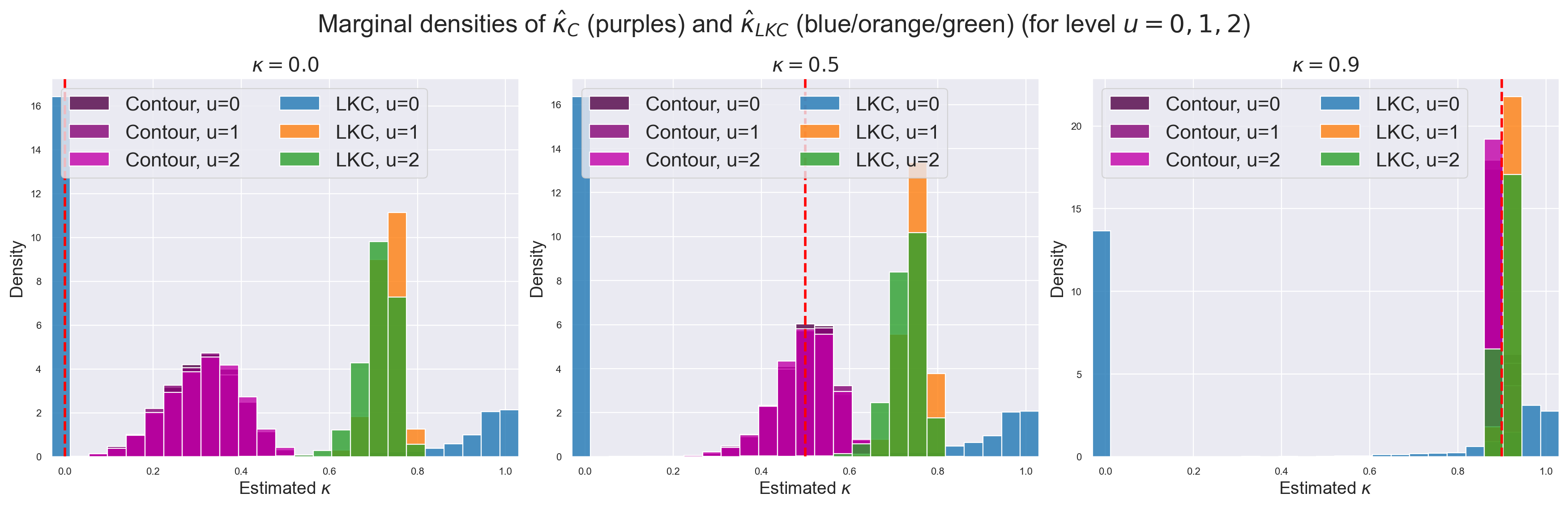}
     \caption{Joint and marginal empirical distribution of $(\widehat\kappa_{ {\rm LKC}},\widehat \kappa_{{\rm C}})$.}
     \label{fig:density_kappa_estimates_MB}
     \label{fig:marginal_density_kappa_estimates_MB}
 \end{figure}

 \vspace*{-0.5cm}
\paragraph{Bias and variance analysis.} 
The asymptotic normality of the underlying functionals (see Theorem~\ref{t:coco} and Corollary~\ref{cor:assymptotic_normality}) implies that the estimators follow a distribution shaped by the inverse of the corresponding elliptic functions. This structural difference leads to distinct behaviors for Contour and~LKC:

\noindent
\textbf{Quasi-isotropy ($\kappa=0$):} A major advantage of the Contour method is that the observed statistic~$\mathcal{F}$ almost surely falls within the domain of the inverse function $g^{-1}(\cdot)$. Consequently, $\widehat\kappa_{\rm C}$ \emph{behaves like a centered normal variable transformed by a smooth function}. In contrast, the LKC statistic often falls outside the admissible domain of the inverse function $R^{-1}(\cdot)$ due to variance in the Euler characteristic estimation. These "impossible" values are truncated to zero (isotropy), creating a large Dirac mass at $\kappa=0$ in the empirical distribution of $\widehat\kappa_{\rm LKC}$. While this artificially lowers the variance, the non-truncated component exhibits a heavy right tail and significant positive bias.

\noindent
\textbf{Anisotropy ($\kappa=0.5, 0.9$):} As the anisotropy increases, the Contour estimator remains unbiased and precise, with a variance that is \emph{remarkably stable across different observation levels $u$}. Conversely, the performance of $\widehat\kappa_{\rm LKC}$ is highly dependent on the level $u$. It is notably unstable at the mean level ($u=0$), where the expected Euler characteristic vanishes, leading to numerical singularities in the ratio $R(\kappa)$. The LKC estimator only achieves accuracy comparable to the Contour method at high excursion levels (e.g., $u=2$), where the geometric approximation of the excursion set becomes more reliable.

\noindent
\textbf{Accuracy:} Overall, $\widehat\kappa_{\rm C}$ provides a \emph{more robust and statistically} efficient measure of anisotropy, particularly close to the mean.
\pagebreak[3]

\subsection{Contour is almost as good as full observation}
We would like to compare the ‘‘Oracle'' which implies observing the gradient field everywhere, against the Contour method which only observes the geometry of a level set. 
In dimension $d=2$, writing $\Lambda=\Var(X'(0))$ with eigenvalues $\lambda_1\ge \lambda_2>0$ and principal direction given by the leading eigenvector, the anisotropy parameter is
\[
\kappa\;=\;\sqrt{1-\lambda_2/\lambda_1},\qquad \theta_0\;=\;\arg(\text{eigenvector associated with }\lambda_1).
\]
In practice (as implemented in the notebooks of COMETE\footnote{\url{https://github.com/ydecastro/COMETE-Contour-method-Gaussian-Random-Fields/}}), we proceed as follows: $(i)$ Compute finite-difference gradients on the $T\times T$ grid; $(ii)$ Form the \emph{empirical variance-covariance matrix} $\widehat\Lambda$ of these gradients over all valid pixels and compute its eigen-decomposition; $(iii)$ Estimate $\widehat\kappa_{\rm grad}^2={1-\hat\lambda_2/\hat\lambda_1}$ and the direction by the leading eigenvector.

This is an \emph{oracle benchmark}: it requires the full field values (not just a level set). It is therefore not applicable in our observation model, but it provides a performance benchmark for what could be achieve with complete data (full information).

\begin{figure}[!ht]
    \centering
    \includegraphics[width=\linewidth]{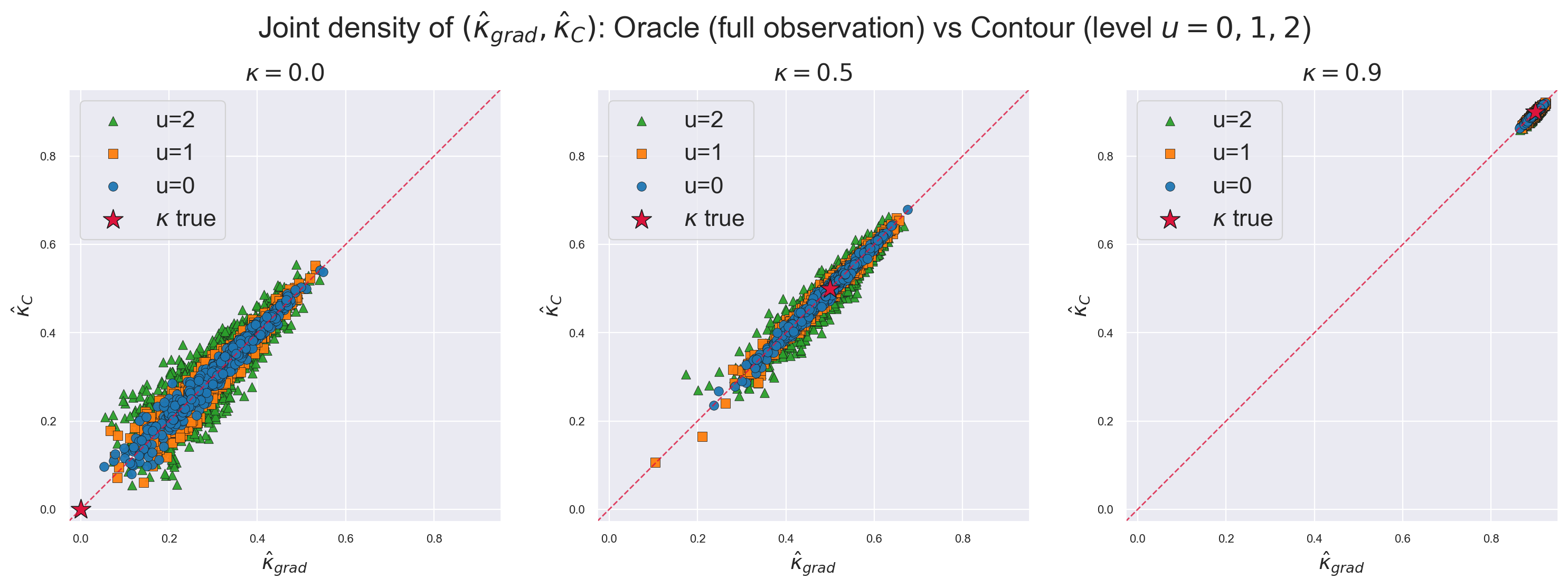}
    \caption{Joint distribution  of $ ( \hat \kappa_{grad},\hat \kappa_C)$ across $\kappa\in\{0,0.5,0.9\}$ and levels $u\in\{0,1,2\}$.} 
    \label{fig:fullobs_cabana}
\end{figure}

The results are presented in Figure~\ref{fig:fullobs_cabana}. It shows that the performance of both estimators is almost indistinguishable and, moreover, that they behave in the same direction: when, for example, one estimator over-estimates, the second does the same. We find that the contour-based estimator $\widehat\kappa_{\rm C}$ is essentially as accurate as the full-observation estimator $\widehat\kappa_{\rm grad}$, while requiring only the level set and remaining agnostic to the field’s mean, variance, and covariance. This supports the practical relevance of the contour approach in settings where full observation is unavailable.

In private simulations, we have compute the normalized gradient covariance matrix on a fixed grid (as in the Oracle benchmark). We computed the anisotropy parameter estimate $\hat\kappa_{\rm nGrad}$ by inverting some elliptic function. We found that $\hat\kappa_{\rm nGrad}$ and $\hat\kappa_{\rm grad}$ are very close and that the joint distribution of $(\hat\kappa_{\rm nGrad},\hat\kappa_{\rm C})$ is almost the same as $(\hat\kappa_{\rm grad},\hat\kappa_{\rm C})$, shown in Figure~\ref{fig:fullobs_cabana}. Hence the normalized gradient convey as much information as the gradient observed on the full domain, when observed along a level set ($\hat\kappa_{\rm C}$) or on the full domain ($\hat\kappa_{\rm nGrad}$).
 
\subsection{Power studies}
\label{sec:Power}

We compare three procedures: MB-Contour (model-based contour test), MB-LKC (model-based Lipschitz–Killing curvatures), and our $\chi^2(2)$-Contour test (model-agnostic). Under the null ($\kappa=0$), the empirical CDF of the $\chi^2(2)$-Contour $p$-values closely follows the uniform CDF across levels $u\in\{0,1,2\}$, indicating excellent calibration (without simulating the null model, model-agnostic test). Under the alternative ($\kappa=0.5$), the $\chi^2(2)$-Contour and MB-Contour tests exhibit the largest power among the three. However, the MB-Contour requires the true model to simulate the null distribution. 

We display the MB-Oracle (Model-based gradient full observation) curve as the oracle reference in Figure~\ref{fig:power}. The MB-LKC test can be clearly less powerful in this setting: First, it suffers from numerical instability when estimating higher-order LKC by counting components/holes which is sensitive and less stable than normal-based implementations, especially at moderate levels and finite windows; Second these experiments suggests that its statistical variance is greater than of Contour methods, resulting in a lack of power (a phenomenon already observed in Figure~\ref{fig:density_kappa_estimates_MB}).

 \begin{figure}[!ht]
    \centering
    \includegraphics[width=0.98\linewidth]{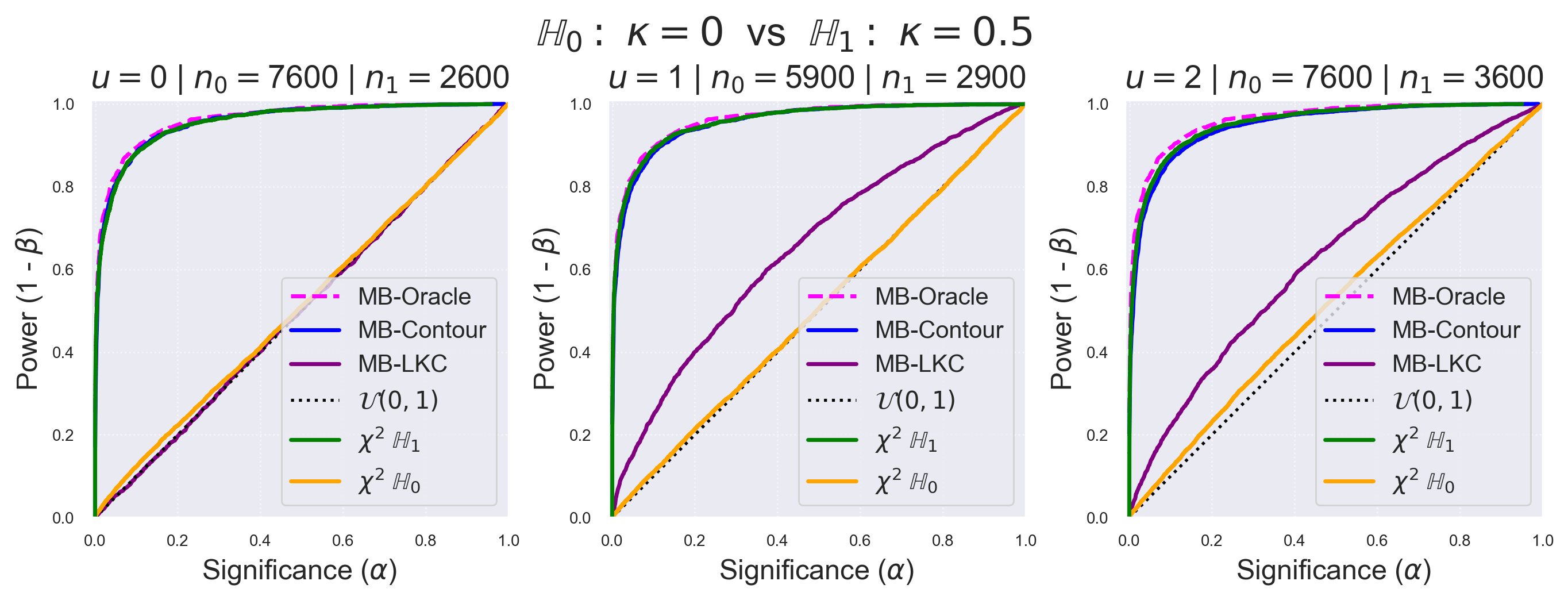}
    \caption{Comparison of the three tests: MB-Contour (Model-Based contour method); MB-LKC (Model-Based Lipschitz-Killing curvatures method); and $\chi^2(2)$-Contour (test based on contour method, introduced in this paper). Note that model based tests require knowing the true model to simulate the null distribution while the $\chi^2(2)$-Contour test is agnostic to the model and does not require its knowledge. The stationary Gaussian random field of unit variance was observed on a window of size $T\times T$ with $T=200$ and the tests were performed on contours at levels $0$, $1$, and $2$. The null hypothesis $\mathds{H}_0$: ‘‘$\kappa=0$'' means that the random field is quasi-isotropic. The distribution under the null is simulated over more than $n_0=6,000$ samples in model-based tests. The power is given when the true law is drawn with respect to the alternative $\mathds{H}_1$:~‘‘$\kappa=0.5$'' for which the random field is anisotropic. The distribution under the alternative is simulated over more than $n_1=2,500$ samples. Under the null, the $\chi^2(2)$-Contour CDF (orange curve) matches the uniform CDF, the test is well calibrated. Under the alternative, the $\chi^2(2)$-Contour test (violet curve) has the best power, significantly exceeding the one of MB-LKC test.}
    \label{fig:power}
\end{figure}

We also observed the expected trends: $(i)$ \emph{increasing the observation window~$T$ improves power} (in private simulations); $(ii)$ the test’s power varies across levels $u$ because contour length and local geometry change but \emph{remains remarkably stable for contour methods (MB and $\chi^2(2)$)} (as also documented in Figure~\ref{fig:density_kappa_estimates_MB}); $(iii)$ choosing the partition size $N$ to balance variance (more cells) and bias (enough points per cell) stabilizes the variance estimator $\hat V^2$ and thus the $\chi^2$ calibration. We choose $N=10$ for levels $u=0,2$ and $N=25$ for level $u=1$. The theoretical study emphasizes that $N=o(T)$ is required in the asymptotic $T\to\infty$. Further theoretical and methodological investigations are necessary to understand how to choose $N$ for finite values of $T$. In this first set of experiments on this subject, we choose $N$ so as to be calibrated under the null. We observed that taking a larger $N$ results in conservative tests (in private simulations), enforcing a controlled Type-$I$ error. 

Overall, the contour-based chi-square approach is both robust (no model knowledge and stable across levels) and competitive in power, while MB procedures are sensitive to model specification, implementation details (LKC estimation) and nuisance parameters (LKC estimation).
 
\subsection{Planck experiment}
\label{sec:Planck}

Using a notebook available on Github\footnote{COMETE: \url{https://github.com/ydecastro/COMETE-Contour-method-Gaussian-Random-Fields/}}, we replicate our pipeline on Planck DR3 CMB (Cosmic Microwave Background) temperature data (Commander product). We extract a rectangular patch of size $1{,}550\times1{,}932$ pixels and proceed as follows: (i) we estimate a representative level set (for display we show $u=175$; the isotropy test itself does not require fixing $u$), (ii) along this level set we compute unit normals and the trigonometric summaries $\cC$ and $\cS$, (iii) we partition the image into a $10\times10$ grid and estimate the variance $\hat V^2$ by sub-sampling as in~\eqref{e:V}, and (iv) we form the chi-square statistic $Q_{n,N}$ in~\eqref{e:qn} and its p-value $\alpha=F_{\chi^2(2)}(Q_{n,N})$ (Theorem~\ref{thm:pvalue}). The principal direction is estimated from $\tfrac12\operatorname{atan2}(\cS,\cC)$. On this patch we reject quasi-isotropy with observed significance $6.57\times10^{-7}$ and estimate the anisotropy angle at $1.74$ radians. We observe similar decisions across nearby display levels and subwindows, suggesting robustness of the conclusion. We stress that our procedure is model-agnostic and uses only contour geometry; a dedicated CMB analysis could additionally account for masking, foreground residuals, and inhomogeneous noise.

\begin{figure}[!ht]
    \centering
    \includegraphics[width=0.45\linewidth]{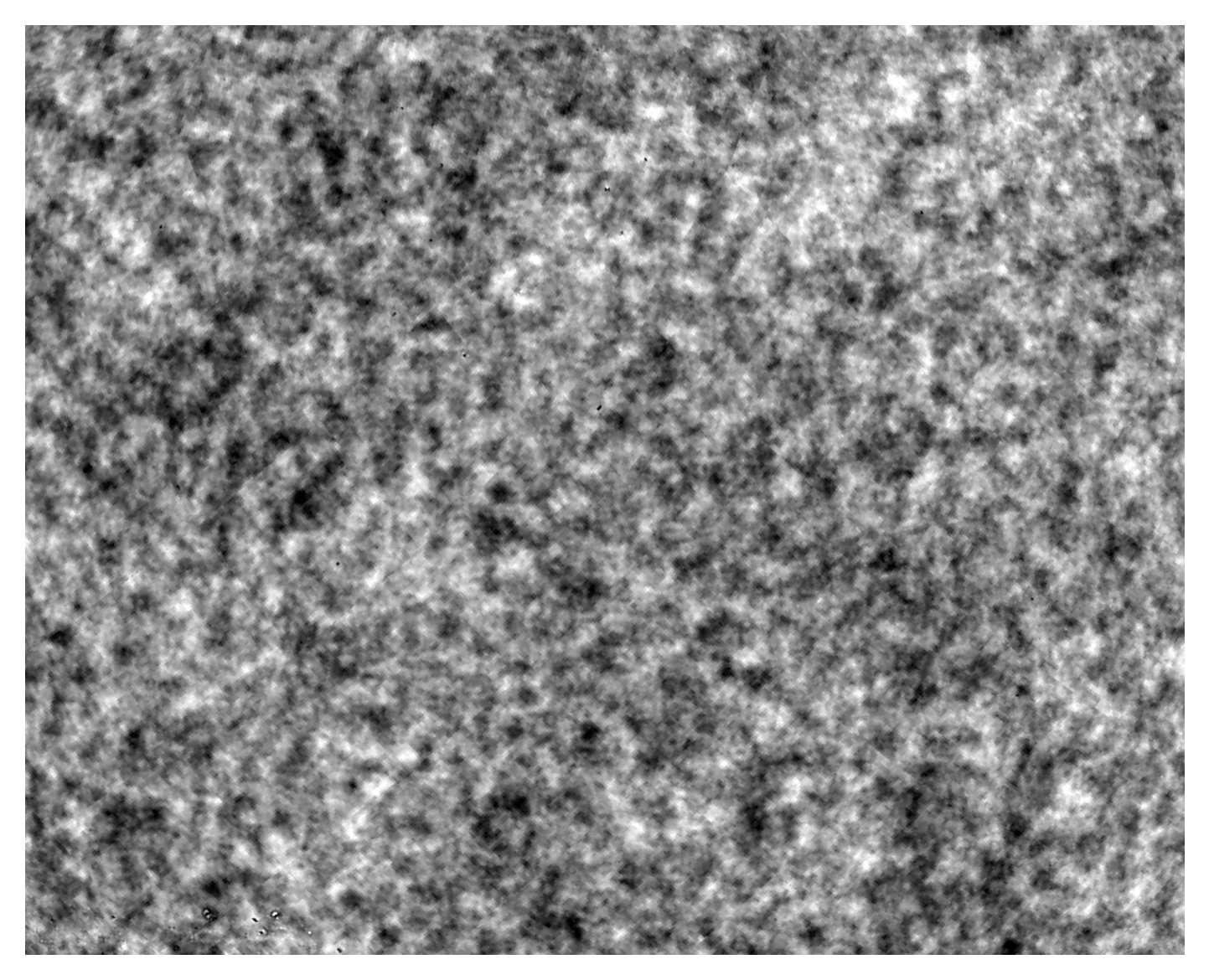}
    \includegraphics[width=0.45\linewidth]{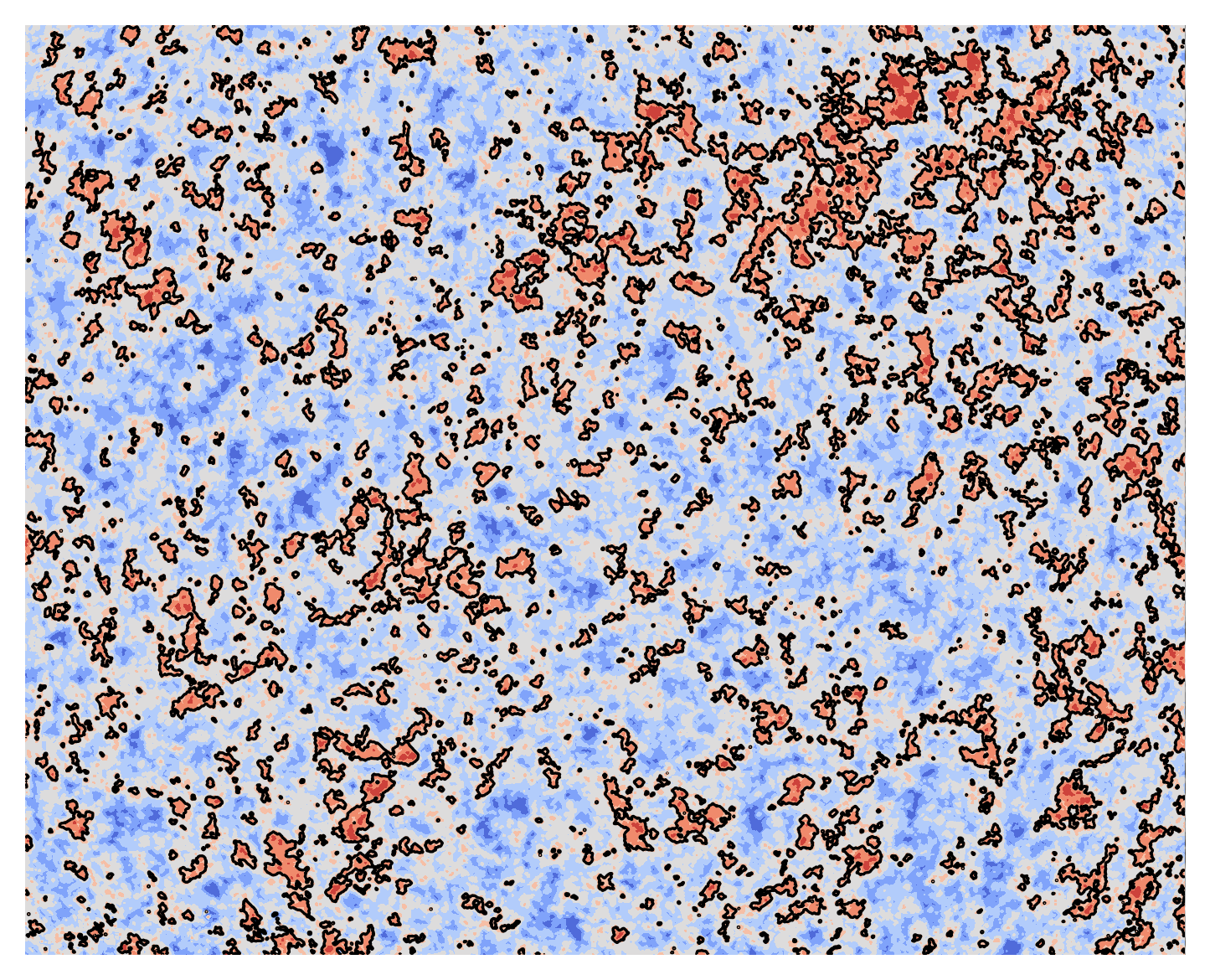}
    \caption{\underline{Left:} $1,550\times 1,932$ image extracted from the Planck data release 3 CMB-only temperature produced by Commander \citep{ESA:CMBMaps:2018}. 
    \underline{Right:} Estimated level set $(X(t)=175)$, where the level $u=175$ has been chosen only for display results.
    \underline{Result:} We performed the $\chi^2(2)$-Contour method on a $10\times 10$ partition. 
    The null hypothesis (the field is quasi-isotropic) has been rejected with an observed significance of $6,57\times 10^{-7}$ and the anisotropy angle has been estimated to $1.74$ radians.
    }
    \label{fig:planck_experiment}
\end{figure}

\subsection{Discussion}

Our numerical study highlights several key distinctions between the contour-based approach and the Lipschitz-Killing curvature (LKC) method.

\paragraph{Estimation performance and the anisotropy angle.}
A unique advantage of the Contour method is its ability to estimate the direction of anisotropy $\theta_0$ via the argument of the vector $(\mathcal{C}, \mathcal{S})$. As shown in Figure \ref{fig:Cabana_MB}, this estimation is robust and precise. Regarding the anisotropy  parameter $\kappa$, the Contour estimator~$\hat{\kappa}_{\mathrm{C}}$ generally outperforms the LKC estimator $\hat{\kappa}_{\mathrm{LKC}}$ in terms of variance, particularly near the mean level ($u \approx \mu$). The LKC method tends to get more precise at high excursion levels ($|u-\mu| \gg \sigma$), where the asymptotic approximations of the Euler characteristic become exact. However, at moderate levels, the LKC method suffers from numerical instability if the Euler characteristic is computed by counting connected components and holes (topological approach). We found that implementing LKC via the integration of the Gaussian curvature density (geometric approach, Eq. \eqref{e:hatec}) significantly improves stability, though it remains less powerful than the Contour method in our simulations.

\paragraph{Independence and combination.}
The estimators $\hat{\kappa}_C$ and $\hat{\kappa}_{LKC}$ are derived from fundamentally different geometric summaries—normals for the former and curvature for the latter. Empirically, we observed that their errors are weakly correlated. While this suggests that combining them (e.g., via the linear combination proposed in Appendix \ref{s:combi}) could theoretically reduce variance, we found that the optimal weights are highly sensitive to the model and the observation level. Consequently, a simple combination is difficult to calibrate in practice, and we recommend using the Contour estimator as the primary tool due to its robustness.

\paragraph{Calibration of the variance estimator.}
The model-agnostic $\chi^2(2)$-Contour test relies on the spatial variance estimator $\hat{V}^2$, which depends on the number of partition cells $N^2$. The choice of $N$ involves a bias-variance trade-off: a large $N$ increases the number of samples for variance estimation but reduces the size of each cell, making the local contour integral approximations less accurate. In our experiments, we found that $N=10$ (100 cells) was appropriate for $u=0$ and $u=2$, while $N=25$ performed better for $u=1$. In practice, we advise practitioners to test stability across a range of $N$. In cases of doubt, choosing a smaller $N$ (fewer, larger blocks) tends to yield a more conservative test, preserving the Type~I error rate.

\paragraph{Comparison to  full information.}
The incomplete observation framework  causes almost no loss   of information with respect to  full observation. 
\paragraph{Conclusion.}
Overall, the Contour method proves to be a  flexible and powerful tool. It is computationally efficient, provides the anisotropy angle, and yields a test that matches or exceeds the power of model-based alternatives without requiring knowledge of the underlying covariance structure. 

\newpage

\bibliographystyle{plainnat}
\bibliography{2024_anisotropie_bibliography}

\newpage

\appendix    
\section {Appendix} \label{s:appendix}

\subsection{Affine processes} \label{s:affine}
Originally, \citet{cabana1987affine} and \citet{wschebor} used  non-Gaussian models.
The null corresponded to a strictly isotropic random field and the alternative hypothesis corresponded to an {\it affine process} as defined below.

\begin{definition}
    The random field $\{ X(t), t\in \cT \subset \R ^2 \to \R\}$ is an affine process if there exist a stationary isotropic random field $\{ Y(t), t\in \cT \to \R\}$ and a $2\times 2$ symmetric positive definite matrix $A$ such that 
\[
X(t) = Y(A t)\,.
\]
\end{definition}

Considering  non-Gaussian  models   implies   the use of complicated mixing assumptions as  $\eta$ dependence. Moreover, the computation of moments through  Kac-Rice formula  demands heavy conditions, see for example, \cite{berzin}.
Explicit examples of random fields satisfying these conditions are lacking.  For  these reasons  
we prefer to use a Gaussian model which is coherent with \citep{berzin} \citep {estrade} and \citep{bierme}. 

       \subsection {Bulinskaya lemmas} \label{s:buli}

\begin{lemma} \label{l:buli1} \cite[Proposition 1.1]{armentano2025general}
Let $X(\cdot)$ be a real valued random field  with $\mathcal C ^1$-paths  defined on an open subset $\cT$ of $\R^d$, and let $u\in \R$. Assume that the
 density $p_{X(t)}$  of $X(t)$  satisfies
 \[
 \int_\cT p_{X(t)} dt < C \mbox{ for all $v$ in some neighbourhood of~$u$ }
\] 
\noindent
Then
\[
 \mathcal H_{d-1}(\{t\in \cT : X(t) = u, X'((t)) = 0\}) = 0, \mbox{ a.s.}\,.
\]   
\end{lemma}

\medskip

\begin{lemma} \label{l:buli2} \cite[Proposition 6.12]{marioluigi}
   Let $X(\cdot)$ be a real valued random field with $\mathcal C ^2$-paths defined on a  open subset $\cT$ of $\R^d$, and let $u \in \mathds R$. Assume that $(X(t),X'(t))$ admits a  joint  density  which is uniformly bounded for $t\in \cT$. Then
   \[
   \{t\in \cT : X(t) = u, X'((t)) = 0\} =\emptyset \mbox{ a.s.}
   \]
\end{lemma}
\subsection {Some calculations} \label{s:palm:calcul}  
To avoid ambiguities we recall the definition of the complete elliptic integral of the first, second and third  kind
\begin{align*}
    K(\kappa) & := \int_0^{\pi/2} 
\big(1- \kappa^2 \sin^2(\theta) \big)^{-1/2}\mathrm d\theta; \\
E(\kappa) &:= \int_0^{\pi/2} \big(1- \kappa^2 \sin^2(\theta) \big) ^{1/2}\mathrm d\theta;  \\
\Pi(n,k) &:= \int_0^{\pi/2} \frac{1}{(1-n\sin^2 \theta) \sqrt{1- \kappa^2 \sin^2(\theta)}} \d \theta.
\end{align*}
 Define 

 \begin{equation} \label{e:c:c}
             \mathds C(\kappa)
        := \frac
                {\int_{-\pi}^\pi \cos^2\!\theta\big(1-\kappa^2\cos^2\!\theta\big)^{-\frac{3}{2}}\mathrm d\theta}
                {\int_{-\pi}^\pi \big(1-\kappa^2\cos^2\!\theta\big)^{-\frac{3}{2}}\mathrm d\theta}.\\
                \end{equation} 
   Computations show that  the numerator is equal to 
   \begin{equation} \label{e:chatgpt}
   \frac4\kappa  \frac{\partial K(\kappa)}{\partial \kappa}   . 
    \end{equation}   
     Using a formula of \citep{whittaker2020course} 
      $$
\frac{\partial K(\kappa)}{\partial \kappa}
=  \frac{E(\kappa) }{\kappa (1-\kappa^2)} - \frac{K(\kappa)}{\kappa} .
      $$
      As for the denominator, by definition it is equal to $ \Pi(\kappa^2,\kappa)$.
Eventually   
 \[
 \mathds C(\kappa) =  \frac{ 4
  E(\kappa) -K(\kappa)(1-\kappa^2)
 }{
  \kappa^2 \Pi(\kappa^2,\kappa)
}.
 \]
 On the other hand 
  \[
   \frac{\E(\cC) }{\E(|\lev|)} = \EPalm( \cos(2\Theta) )
   = \frac{ 
   \int_{-\pi}^\pi \cos(2(\theta+\theta_0)) \big(1-\kappa^2\cos^2\!\theta\big)^{-\frac{3}{2}}\mathrm d\theta}
   {\int_{-\pi}^\pi \big(1-\kappa^2\cos^2\!\theta\big)^{-\frac{3}{2}}\mathrm d\theta}.
  \]
 Using $\cos(2(\theta +\theta_0)) = \cos(2\theta_0) (2\cos^2(\theta)-1) - \sin(2\theta_0) \sin(2\theta) $ and remarking that the term in $\sin(2\theta)$  give a null contribution, we get 
 \[\EPalm \cos(2\Theta)  = \cos(2\theta_0) (2  \mathds C(\kappa) -1) \]
Using $\sin(2(\theta +\theta_0)) =  \cos(2\theta_0) \sin(2\theta) +  \sin(2\theta_0)   (2\cos^2(\theta)-1) $, we obtain in the same fashion:
 \[\EPalm \sin(2\Theta)  = \sin(2\theta_0) (2  \mathds C(\kappa) -1) \]
 As a consequence the function $g(\cdot)$  defined by \eqref{e:cabana}
 takes the value 
 \[
               g(\kappa) = (2  \mathds C(\kappa) -1) = 
                \frac{ 8
  E(\kappa) -2K(\kappa)(1-\kappa^2)- \kappa^2 \Pi(\kappa^2,\kappa)
 }{
  \kappa^2 \Pi(\kappa^2,\kappa)
}
\]        
Using an identification method   $\theta_0$ and $\kappa$ are estimated by \eqref{e:theta:hat}-\eqref{e:k:hat}.

It remains to give the proof of Lemma \ref{l:g} i.e. that $g(\cdot) $ is invertible.

Set
$C_i:= \kappa_i^{-3/2}$, $ i=1,2$  and, using homogeneity of the problem,  suppose that $C_1 = 1+a, C_2 = 1-a$. We have 
\[
 \mathds C(\kappa)
    = \frac
                {\int_{-\pi}^\pi \cos^2\!\theta\big(1+ a \cos(2\theta)\big)^{-\frac{3}{2}}\mathrm d\theta}
                {\int_{-\pi}^\pi \big(1+ a \cos(2\theta)\big)^{-\frac{3}{2}}\mathrm d\theta}\\
    =   \frac
                {\int_{-\pi}^\pi  \frac{1 + \cos(\phi)}{2} \big(1+ a \cos(\phi)\big)^{-\frac{3}{2}}\mathrm d\phi}
                {\int_{-\pi}^\pi \big(1+ a \cos(\phi)\big)^{-\frac{3}{2}}\mathrm d\phi} =: \frac{U}{V},   
\]
setting $\phi= 2\theta$.
By differentiation with respect to  $a$, the numerator of the derivative is given by 
$
U'V-UV'$ 
With 
\begin{align*}
U' &=  -3/2   \int_{-\pi}^\pi  \frac{1 + \cos(\phi)}{2}    \cos(\phi) \big(1+ a \cos(\phi)\big)^{-\frac{5}{2}}\mathrm d\phi\\
V' &=  -3/2   \int_{-\pi}^\pi     \cos(\phi') \big(1+ a \cos(\phi')\big)^{-\frac{5}{2}}\mathrm d\phi'
\end{align*}
So that 
\begin{align*}
   &U'V-UV' =
   \\
   &-\frac32
\Big( \int_{-\pi}^\pi \int_{-\pi}^\pi  (1+a\cos\phi)  (1+ a\cos\phi')  \big( \cos^2\phi -\cos\phi \cos\phi'   \big)  \big(1+ a \cos\phi\big)^{-\frac{5}{2}}\big(1+ a \cos\phi'\big)^{-\frac{5}{2}}\mathrm d\phi\mathrm d\phi' \Big) 
\end{align*}
which is non-positive by the Cauchy-Schwarz inequality.  

\section{Contour is consistent}

\subsection{Covariances give the anisotropy parameters and directions}
Based on the derivations of Theorem~\ref{thm:Palm_normal}, we know that the covariance of the normalized gradient shares the same structural information as the covariance of the full gradient. The fundamental reason there is no loss of information regarding anisotropy is that \textbf{both matrices share the exact same eigenvectors}, and their eigenvalues are linked by a strictly invertible mapping. Recall $\Lambda$ denote the covariance of the full gradient and let $\Sigma_{\PnPalm}$ denote the covariance of the normalized gradient (under the normalized Palm distribution). One can prove (see Proof of Theorem~\ref{thm:Palm_normal}):

\begin{subequations}
\begin{itemize}
    \item \textbf{Full Gradient:} 
    \begin{equation}
        \Lambda = {\bf P}^\top \text{Diag}(\vv\kappa^2) {\bf P}\propto \int_{\mathbb{S}^{d-1}} \zeta \zeta^\top \left(\zeta^\top \Lambda^{-1} \zeta\right)^{-\frac{d+2}{2}} d\eta(\zeta)
    \end{equation}
    \item \textbf{Normalized Gradient along a level set:} 
    \begin{equation}
        \Sigma_{\PnPalm} = {\bf P}^\top \text{Diag}(\mathcal{Z}(\vec{\kappa})) {\bf P}\propto \int_{\mathbb{S}^{d-1}} \zeta \zeta^\top \left(\zeta^\top \Lambda^{-1} \zeta\right)^{-\frac{d+1}{2}} d\eta(\zeta)
    \end{equation}
\end{itemize}
\end{subequations}

Because the rotation matrix ${\bf P}$ (which defines the directions of anisotropy) is identical in both equations, observing the normalized gradient provides the exact same directional information as the full gradient. The eigenvalues $\mathcal{Z}(\vec{\kappa})$ differ from $\vec{\kappa}^2$, but as proven in Theorem~\ref{thm:inverse_palm}, the mapping from~$\mathcal{Z}$ back to~$\vec{\kappa}$ is unique and invertible via a convex optimization problem. Comparing the two final integrals reveals why the normalized gradient retains the anisotropy information:

\begin{itemize}
    \item \textbf{Similarity:} The only difference between the full gradient covariance and the normalized gradient covariance is the power of the scaling term inside the integral:
    \begin{itemize}
        \item Full Gradient weight: $\left(\zeta^\top \Lambda^{-1} \zeta\right)^{-\frac{d+2}{2}}$
        \item Normalized Gradient weight: $\left(\zeta^\top \Lambda^{-1} \zeta\right)^{-\frac{d+1}{2}}$
    \end{itemize}
    The term $\zeta^\top \Lambda^{-1} \zeta$ contains the ‘‘shape'' of the anisotropy (the ellipse defined by $\Lambda$). Both integrals are heavily weighted by this shape factor.
    
    \item \textbf{Identical Principal Directions:} Because the weighting functions $(\zeta^\top \Lambda^{-1} \zeta)^{-p}$ are symmetric with respect to the axes of the ellipsoid defined by $\Lambda$, the resulting integral matrices \textit{must} be diagonalized by the same eigenbasis ${\bf P}$.
    
    \item \textbf{Recoverability:} While the eigenvalues $\mathcal{Z}(\vec{\kappa})$ are numerically different from $\vec{\kappa}$ due to the different exponent in the integral, we prove that the map $\vec{\kappa} \to \mathcal{Z}(\vec{\kappa})$ is a diffeomorphism (Theorem~\ref{thm:inverse_palm}). There is a unique correspondence between the ‘‘shape'' observed in the normalized gradient covariance and the ‘‘shape'' of the full gradient covariance.
\end{itemize}

\subsection{Definition of the Sample Statistic}
We identify that the normalized Palm distribution $\PnPalm$ is defined via expectations, which are unknown in practice. To construct a valid statistic, we must remove the expectation and consider the \textbf{sample version} computed from the single observed realization of the level set $\mathcal{L}(\mathcal{T})$.

We define the sample statistic $\widehat{\Sigma}_{{\rm nPalm}}$ as the empirical covariance of the normalized gradients (normal vectors) along the observed level set:
\begin{equation}
\notag
    \widehat{\Sigma}_{{\rm nPalm}} := \frac{1}{|\mathcal{L}(\mathcal{T})|} \int_{\mathcal{L}(\mathcal{T})} N(t) N(t)^\top d\mathcal{H}^{d-1}(t)
\end{equation}
where $N(t) = \frac{X'(t)}{||X'(t)||}$ is the unit normal vector at position $t$ on the level set. This quantity is fully observable from the data (the image).

\subsection{Proof of consistency}
The proof of consistency follows mutatis mutandis the $2$ dimensional case in \citep[Th. 4.1]{berzin}. We sketch it here. 
Under assumption \eqref{eq:assumption} and $r(t)\to0$ as $\|t\|\to\infty$, \citep[Th.6.5.4]{adler} implies that the stationary Gaussian field $X(\cdot)$ is ergodic. By \citep[Th.6.5.2]{adler} this property is inherited by the field $\xi(t,s):=\cC_{ij}([t,t+1]\times[s,s+1])$ which integrated on $\cT_n$ yields $\cC_{ij}(\cT_n)$ up to negligible terms, see \citep[L.4.2]{berzin}. The ergodic Theorem \citep[Th. 6.5.1]{adler} implies that the spatial average (the sample statistic $\widehat{\Sigma}_{{\rm nPalm}}$) converges to the ensemble average (the normalized Palm expectation) as the observation window grows ($\mathcal{T}_n \nearrow \mathbb{R}^d$ as $n\to\infty$):
\begin{equation}
\notag
    \widehat{\Sigma}_{{\rm nPalm}} \xrightarrow[n \to \infty]{a.s.} \mathbb{E}_{{\rm nPalm}}[N(t)N(t)^\top] =: \Sigma_{{\PnPalm}}(\vec{\kappa})
\end{equation}
This establishes that the sample statistic is a consistent estimator of the theoretical moment matrix~$\Sigma_{\PnPalm}$.

We know that the mapping from the model parameters to the expected statistics is a \textbf{diffeomorphism} (a smooth, invertible bijection). 
Let $\Psi$ be the function mapping anisotropy parameters to the Palm covariance:
\begin{subequations}
\begin{equation}
    \Psi: \vec{\kappa} \mapsto \Sigma_{\PnPalm}(\vec{\kappa})
\end{equation}
Theorem~\ref{thm:inverse_palm} proves that $\Psi$ is strictly invertible. Therefore, the inverse function $\Psi^{-1}$ exists and is unique:
\begin{equation}
    \vec{\kappa} = \Psi^{-1}(\Sigma_{\PnPalm}) = \lim_n \Psi^{-1}(\widehat{\Sigma}_{{\rm nPalm}})
\end{equation}
by the continuous mapping theorem. Hence, the empirical covariance $\Psi^{-1}(\widehat{\Sigma}_{{\rm nPalm}})$ is a consistent statistic for the anisotropy $\vec{\kappa}$.
\end{subequations}

\section{Combining contour and LKC}  \label{s:combi}
The aim of  this subsection is to propose a synthesis between the estimator of \citet{cabana1987affine} and \citet{bierme}. 
 
If we except the Wschebor method, which is a variant,  the Cabaña  method is the only one  that gives  an estimator of $\theta_0$. So we focus on the  measure of anisotropy and in particular the estimation of the parameter $\kappa$.
From  \eqref{e:hatp} and \eqref{e:hatec} we have  that 
   \begin{align*}
   \frac{\hat {GC}} {(\widehat P)^2}   &\simeq R(\kappa),
    \\
   \cF
   &\simeq g(\kappa),
  \end{align*}
where $\cF$ is defined in \eqref{e:cf} and $\simeq$ means approximative equality due to the statistical variability.
This can be combined  by least square  regression, minimizing in $\kappa$ 
\[
 \alpha_1\Big| \frac{ \hat {GC}} {\widehat P^2}  - R(\kappa)\Big|^2 + \alpha_2
\Big|\cF-g(\kappa)\Big|^2, \quad 0<\alpha_1, \alpha_2 <1, \ \alpha_1 + \alpha_2  =1.
\]
The obtained estimator will be denoted $\widehat\kappa_c$
The determination  of optimal weights  $ \alpha_1, \alpha_2$ is an open problem and            
the  proposed solution is likely suboptimal. Of course   there are  many other ways of combining information. But we stress that, anyway,  it is better than using only a part of the information.

\section{A note on the Almond curve}
\label{sec:almond}
The paper \citet{bierme} introduces a nuisance-parameter free representation (the ``almond'' curve)
\[
    \mathcal A_{\kappa} := \Big\{ (x,y)\in(0,1] : y^2 + \tfrac{64}{\pi^3} R(\kappa)^2 x^2 \log x = 0 \Big\},
\]
with coordinates $(\tilde x, \tilde y)$ in our notation
\begin{align*}
    	\tilde x &:= \exp\big(-(\mu-u)^2/(2\sigma^2)\big), \\
    	\tilde y &:= R(\kappa) \frac{\pi^2}{\sqrt{2\pi}} \frac{u-\mu}{\sigma} \exp\big(-(\mu-u)^2/(2\sigma^2)\big).
\end{align*}
While graphically appealing, estimation of $\kappa$ still relies on inverting $R$ from $(\widehat{\mathrm{GC}}, \widehat P)$; plotting $(\tilde x, \tilde y)$ does not by itself yield a consistent estimator nor a direct test, but can help to visualize departure from isotropy: the almond curves $\mathcal{A}_\kappa$ are inside the curve $\mathcal{A}_0$ ($\kappa=0$ and $R(0)=4/\pi^2$, effectively isotropic), see~\citep{bierme} for further details.

\newpage

\section{Table of notation} \label{sec:table_of_notation}
\begin{table}[h!]
\centering
\caption{Main notation used in the paper.}
\label{tab:notations}
\begin{tabular}{ll}
\hline
\textbf{Symbol} & \textbf{Meaning} \\
\hline
$\R,\N,\Z,\C$ & Real, natural, integer, complex numbers. \\
$d,\ \I_d$ & Space dimension; identity matrix of size $d$. \\
$\cT,\ \cT_n=(-n,n)^d$ & Observation window; growing domains. \\
$|\cT|$ & Lebesgue measure (volume) of $\cT$. \\
$\cH^k$ & Hausdorff measure of dimension $k$. \\
$X(\cdot)$ & Stationary Gaussian random field on $\cT$. \\
$Y(t)=f(X(t))$ & Transformed Gaussian field ($f$ monotone $\mathcal C^2$). \\
$\mu,\ \sigma^2$ & Mean and variance of $X(0)$. \\
$u$ & Threshold (unknown level). \\
$(u,\mu,\sigma)$ & Nuisance parameters. \\
$\cE(\cT)=\{t\in\cT:X(t)>u\}$ & Excursion set (binary observation). \\
$\lev(\cT)=\{t\in\cT:X(t)=u\}$ & Level set; $|\lev(\cT)|=\cH^{d-1}(\lev(\cT))$. \\
$X'(t)$ & Gradient of $X$ at $t$. \\
$X''(t)$ & Hessian of $X$ at $t$. \\
$\Lambda=\Var(X'(0))$ & Gradient covariance matrix. \\
${\bm P},\ {\bm D}=\diag(\kappa_1^2,\ldots,\kappa_d^2)$ & Eigenbasis and eigenvalues of $\Lambda$. \\
$\k=(\kappa_1,\ldots,\kappa_d)$ & Model eigenvalues ($\kappa_1\ge\cdots\ge\kappa_d>0$). \\
$\Delta_+=\{\kappa\in\R_+^d:\sum_i\kappa_i^2=1\}$ & Normalized eigenvalue simplex. \\
$d=2:\ {\bm P}_{\theta_0}$ & Rotation matrix by angle $\theta_0$. \\
$d=2:\ \kappa=\sqrt{1-\kappa_2^2/\kappa_1^2}$ & 2D anisotropy parameter in $[0,1)$. \\
$\Theta(t)$ & Gradient angle at $t$ (2D). \\
$N(t)=X'(t)/\|X'(t)\|$ & Normalized gradient (unit normal on $\mathds S^{d-1}$). \\
$\cC(\cT),\ \cS(\cT)$ & Caba\~na contour integrals of $\cos(2\Theta)$, $\sin(2\Theta)$. \\
$\cC_{ij}(\cT)=\int_{\lev(\cT)} N_iN_j\,\d\cH^{d-1}$ & Higher-dim. contour covariance entries. \\
$C_j(\cT)$ & $j$-th Lipschitz–Killing curvature of $\lev(\cT)$. \\
$g(\kappa)$ & Elliptic link: $\E(\cC)/\E(|\lev|)=\cos(2\theta_0)g(\kappa)$. \\
$R(\kappa)$ & LKC-based ratio $=\sqrt{1-\kappa^2}/E(\kappa)^2$. \\
$\widehat\kappa_{\rm C}$ & Contour estimator via $g^{-1}$. \\
$\widehat\kappa_{\rm LKC}$ & LKC estimator via $R^{-1}$ (Eq.~\eqref{e:kappa:LKC}). \\
$\widehat\kappa_{\rm grad}$ & Oracle estimator from gradient covariance. \\
$\PPalm,\EPalm$ & Palm law / expectation of $X'(t)$ on $\lev$. \\
$\PnPalm,\EnPalm$ & Palm law / expectation of $N(t)$ (normalized gradient). \\
$\eta$ & Uniform probability measure on $\mathds S^{d-1}$. \\
$f_\Theta$ & Palm density of angle $\Theta$. \\
$\mathcal Z(\kappa)$ & Palm normalized-gradient eigenvalues vector. \\
$\mathcal Z_\ell(\kappa)$ & $\ell$-th component of $\mathcal Z(\kappa)$. \\
$Q_{n,N}$ & Chi-square test statistic (Eq.~\eqref{e:qn}). \\
$\hat V^2$ & Block variance estimator (Eq.~\eqref{e:V}). \\
$\alpha=F_{\chi^2(2)}(Q_{n,N})$ & P-value (model-agnostic isotropy test). \\
$I_n$ & Scaled contour integral (Eq.~\eqref{eq:non_asymptotic_Palm}). \\
$J_n$ & Centered, rescaled version for CLT. \\
$K(\kappa),E(\kappa),\Pi(n,\kappa)$ & Complete elliptic integrals (App.~\ref{s:palm:calcul}). \\
$\phi,\Phi$ & Standard Gaussian pdf and cdf. \\
$\chi^2(2)$ & Chi-square law with 2 degrees of freedom. \\
\hline
\end{tabular}
\end{table}

\end{document}